\documentclass{amsart}
\usepackage{amssymb}
\usepackage{amsmath}
\usepackage{amsfonts}

\newtheorem{theorem}{Theorem}
\theoremstyle{plain}

\newtheorem{corollary}{Corollary}

\newtheorem{definition}{Definition}

\newtheorem{lemma}{Lemma}
\newtheorem{notation}{Notation}

\newtheorem{proposition}{Proposition}
\newtheorem{remark}{Remark}

\numberwithin{equation}{section}
\input{tcilatex}

\begin{document}
\title[]{Global Existence and Long-time Behaviour of Nonlinear Equation of
Schr\"{o}dinger type }
\author{Kamal N. Soltanov}
\address{ {\small Dep. Math. Fac. Sci. Hacettepe University, Beytepe,
Ankara, TR-06800, TURKEY ; }\\
{\small Tel. +90 312 2977860; Fax: +90 312 2992017\ }}
\email{soltanov@hacettepe.edu.tr ; sultan\_kamal@hotmail.com\ }
\urladdr{http://www.mat.hacettepe.edu.tr/personel/akademik/ksoltanov/index.html }
\date{}
\subjclass[2000]{Primary 35J60, 35B33, 35Q55 ; Secondary 33E30, 46F10, 46T20 
}
\keywords{Nonlinear Schr\"{o}dinger equation, generalized function
coefficient, general existence theorem, solvability theorem, behaviour of
solution }

\begin{abstract}
In this paper we study a mixed problem for the nonlinear Schr\"{o}dinger
equation globally that have a nonlinear adding, in which the coefficient is
a generalized function. Here is proved a global solvability theorem of the
considered problem with use of the general solvability theorem of the
article [30]. Furthermore here is investigated also the behaviour of the
solution of the studied problem.
\end{abstract}

\thanks{Supported by 110T558-projekt of TUBITAK}
\maketitle

We consider the following problem for the nonhomogeneous nonlinear Schr\"{o}%
dinger equation 
\begin{equation}
i\frac{\partial u}{\partial t}-\Delta u+f\left( x,u\right) =h\left(
t,x\right) ,\quad \left( t,x\right) \in R_{+}\times \Omega \equiv Q, 
\tag{0.1}
\end{equation}%
\begin{equation}
u\left( 0,x\right) =u_{0}\left( x\right) ,\quad x\in \Omega \subset R^{n},\
n\geq 1;\quad u\left\vert _{~R_{+}\times \partial \Omega }\right. =0, 
\tag{0.2}
\end{equation}%
where $h\left( t,x\right) $ and $u_{0}\left( x\right) $ are complex
functions, $f\left( x,\tau \right) $ is a distribution (generalized
function) with respect to variable $x\in \Omega $, $\Omega $ is a bounded
domain with sufficiently smooth boundary $\partial \Omega $, $i\equiv \sqrt{%
-1}$. We investigate this problem in the case, when the function $f(x,t)$
can be represented as $f\left( x,u\right) =q\left( x\right) \left\vert
u\left( t,x\right) \right\vert ^{p-2}u\left( t,x\right) +a\left( x\right)
\left\vert u\left( t,x\right) \right\vert ^{\widetilde{p}-2}u\left(
t,x\right) $, i.e. the function $f$ has the growth with respect to unknown
function of the polynomial type, where $a:\Omega \longrightarrow R$ is some
function and $q:\Omega \longrightarrow R$ is a generalized function, $p\geq
2 $, $\widetilde{p}\geq 2$, $h\in L_{2}\left( Q\right) $ (i.e. $h\left(
t,x\right) \equiv h_{1}\left( t,x\right) +ih_{2}\left( t,x\right) $ and $%
h_{j}\in L_{2}\left( Q\right) $, $j=1,2$).

The nonlinear Schr\"{o}dinger equation of the type (0.1), and also
steady-state case of the equation (0.1) arises in several models of
different physical phenomena corresponding to various function $f$. The
equation of such type were studied in many articles under different
conditions on the function $f$ in the dynamic case (see for example [2, 4, 6
- 9, 14, 17, 18, 20, 22, 24, 31, 32, 34] and the references therein) and in
the steady-state case (see, for example, [1, 3, 5, 10 - 16, 18, 19, 23 - 26,
28, 32, 33, 35] and references therein). It is known that in this case the
equation (0.1) in the steady-state case (i.e. if $u$ is independent of $t$)
is an equation of the semiclassical nonlinear Schr\"{o}dinger type (i.e.
NLS) (see, [1, 2, 3, 10, 13] and references therein). Considerable attention
has been paid in recent years to the problem (0.1) for small $\varepsilon >0$
as the coefficients of the linear part since the solutions are known as in
the semiclassical states, which can be used to describe the transition from
quantum to classical mechanics (see, [3, 10 - 14, 23 - 25, 32 - 35] and
references therein).

In the above mentioned articles the equation (0.1), and also the
steady-state case was considered with various functions $f(x,u)$ that are
mainly Caratheodory functions\footnote{%
Let $f:\Omega \times R^{m}\longrightarrow R$ be a given function, where $%
\Omega $ is a nonempty measurable set in $R^{n}$ and $n,m\geq 1$. Then $f$
is Caratheodory function if the following hold: $x\longrightarrow f\left(
x,\eta \right) $ is measurable on $\Omega $ for all $\eta \in R^{m}$, and $%
\eta \longrightarrow f\left( x,\eta \right) $ is continuous on $R^{m}$ for
almost all $x\in \Omega $.} with some additional properties. Moreover in
some of these articles are presumed, that dates of considered problem
possess more smoothness and study the behaviour of a solution of the posed
problem with use the Fourier mod. Although such cases when $f(x,u)$
possesses a singularity with respect to the variable $x$ of certain type
were also investigated (as equations Emden-Fowler, Yamabe, NLS etc.), but in
all of these articles the coefficient $q\left( x\right) $ is a function in
the usual sense (of a Lebesgue space or of a Sobolev space).

In this paper we study problem (0.1)-(0.2) in the case when $f$ have the
above representation and the function $q$ is a generalized function and the
beahviour of the solution. Moreover here for the proof of the existence
theorem of the problem is used some different method, which allow us several
other possibility. It should be noted that the steady-state case of the
problem of such type were studied in [28]. Here we study problem (0.1)-(02)
globally in the dynamical case, which in [31] is studied for $t\in \left(
0,T\right) $, $T<\infty $. Here an existence theorem (section 1) for the
problem (0.1) - (0.2) is proved in the model case when $f\left( x,u\right) $
only has the above expression (section 4).

In section 2 we have defined how to understand the equation (0.1) with use
of representation of certain generalized functions and properties of some
special class of functions (see, for example, [27, 28]). In section 3 we
have conducted variants of the general results from [29, 30], on which the
proof of the solvability theorem is based and in section 5 is studied the
behaviour of the solutions of the considered problem under certain additions
conditions.

\section{Statement of the Main Solvability Result}

Let the operator $f\left( x,u\right) $ have the form 
\begin{equation}
f\left( x,u\right) =q\left( x\right) \left\vert u\right\vert ^{p-2}u+a\left(
x\right) \left\vert u\left( t,x\right) \right\vert ^{\widetilde{p}-2}u\left(
t,x\right)  \tag{$1.1$}
\end{equation}%
in the generalized sense, where $q\in W^{-1,p_{0}}\left( \Omega \right) $, $%
p_{0}\geq 2$ (it should be noted that either $p_{0}\equiv p_{0}\left(
p\right) $ or $p\equiv p\left( p_{0}\right) $), $a:\Omega \longrightarrow R$
and $u:Q\longrightarrow 
\mathbb{C}
$ is an element of the space of sufficiently smooth functions that will be
determined below (see, Section 2). Consequently the function $q\left(
x\right) $ is a generalized function, which has singularity of the order 1.

We will set some necessary denotations. Everywhere later the expression of
the type $u\in L^{m}\left( R_{+};W_{0}^{1,2}\left( \Omega \right) \right)
\cap L^{2}\left( R_{+};W_{0}^{1,2}\left( \Omega \right) \right) $ $\equiv
L^{\left( 2,m\right) }\left( R_{+};W_{0}^{1,2}\left( \Omega \right) \right) $
for $u:Q\longrightarrow 
\mathbb{C}
$ denote the following 
\begin{equation*}
\left( u_{1},u_{2}\right) \in \left( L^{m_{1}}\left( R_{+};W_{0}^{1,2}\left(
\Omega \right) \right) \right) ^{2}\equiv \left( L^{m_{1}}\left(
R_{+};W_{0}^{1,2}\left( \Omega \right) \right) ,L^{m_{1}}\left(
R_{+};W_{0}^{1,2}\left( \Omega \right) \right) \right)
\end{equation*}%
holds for $m_{1}\in \left[ 2,m\right] $, where $u\left( t,x\right) \equiv
u_{1}\left( t,x\right) +iu_{2}\left( t,x\right) $, $m\geq 2^{\ast }$
consequently we can set $u\left( t,x\right) \equiv \left( u_{1}\left(
t,x\right) ,u_{2}\left( t,x\right) \right) $, i.e. $u:Q\longrightarrow R^{2}$%
;

Everywhere later $\left\langle \cdot ,\cdot \right\rangle $ and $\left[
\cdot ,\cdot \right] $ denote the dual form for the pair $\left( X,X^{\ast
}\right) $ of the Banach space $X$ and its dual space $X^{\ast }$ , for
example, in the case when $X\equiv W_{0}^{1,2}\left( \Omega \right) $ and $%
X\equiv L^{m_{1}}\left( R_{+};W_{0}^{1,2}\left( \Omega \right) \right) $ we
have 
\begin{equation*}
\left( X,X^{\ast }\right) \equiv \left( \left( W_{0}^{1,2}\left( \Omega
\right) \right) ^{2},\left( W_{0}^{1,2}\left( \Omega \right) \right)
^{2}\right)
\end{equation*}%
and 
\begin{equation*}
\left( X,X^{\ast }\right) \equiv \left( \left( L^{m_{1}}\left(
R_{+};W_{0}^{1,2}\left( \Omega \right) \right) \right) ^{2},\left(
L^{m_{1}^{\prime }}\left( R_{+};W^{-1,2}\left( \Omega \right) \right)
\right) ^{2}\right) ,
\end{equation*}%
respectively, where $m_{1}^{\prime }=\frac{m_{1}}{m_{1}-1}$. In the other
words we will understand these expressions everywhere later as the following
representations 
\begin{equation*}
\left\langle g,w\right\rangle \equiv \underset{\Omega }{\int }g\left(
x\right) \overline{w}\left( x\right) dx,\ \ g_{j}\in W_{0}^{1,2}\left(
\Omega \right) ,\ w_{j}\in W^{-1,2}\left( \Omega \right) ,\ g\equiv
g_{1}+ig_{2};
\end{equation*}%
and%
\begin{eqnarray*}
\left[ g,w\right] &\equiv &\underset{0}{\overset{\infty }{\int }}\underset{%
\Omega }{\int }g\left( t,x\right) \overline{w}\left( t,x\right) dxdt,\ \
g_{j}\in L^{m_{1}}\left( R_{+};W_{0}^{1,2}\left( \Omega \right) \right) , \\
w_{j} &\in &L^{m_{1}^{\prime }}\left( R_{+};W^{-1,2}\left( \Omega \right)
\right)
\end{eqnarray*}%
respectively.

Assume the following conditions: (\textit{i}) let $\widetilde{p}<\frac{n+2}{%
n-2}$ if $n\geq 3$, $\widetilde{p}\in \left[ 2,\infty \right) $ if $n=1,2$
and $a\in L^{\infty }\left( \Omega \right) $;

(\textit{ii}) there exist numbers $k_{0}\geq 0$, $p_{2}\geq 1$ and $%
k_{1}\leq \min \left\{ 1;\frac{\widetilde{p}}{p}\right\} $ such that $1\leq
p_{2}<\frac{2n}{n-2}$, if \ $n\geq 3$, $2\leq p_{2}<\infty $, if $\ n=1,2$
and 
\begin{equation}
\left\langle a\left( x\right) \left\vert u\right\vert ^{\widetilde{p}-2}u,%
\overline{u}\right\rangle \geq -k_{0}\left\Vert u\right\Vert
_{p_{2}}^{2}-k_{1}\left\langle q\left( x\right) \left\vert u\right\vert
^{p-2}u,\overline{u}\right\rangle  \tag{1.2}
\end{equation}%
holds for any $u\in L^{\left( 2,m\right) }\left( R_{+};W_{0}^{1,2}\left(
\Omega \right) \right) $ and a.e. $t\geq 0$, where $1>C\left( 2,p_{2}\right)
^{2}\cdot k_{0}$\footnote{%
here $C\left( 2;p_{2}\right) $ is the constant of the known inequality of
Embedding Theorems for Sobolev spaces 
\begin{equation*}
\left\Vert \nabla u\right\Vert _{2}\geq C\left( 2;p_{2}\right) \left\Vert
u\right\Vert _{p_{2}},\ \forall u\in W_{0}^{1,2}\left( \Omega \right) .
\end{equation*}%
}.

\begin{definition}
A function 
\begin{equation*}
u\in L^{m}\left( R_{+};W_{0}^{1,2}\left( \Omega \right) \right) \cap
L^{2}\left( R_{+};W_{0}^{1,2}\left( \Omega \right) \right) \cap
\end{equation*}%
\begin{equation*}
\left\{ u\left\vert \ \frac{\partial u}{\partial t}\in L^{2}\left(
R_{+};L^{2}\left( \Omega \right) \right) \right. ;\ u\left( 0,x\right)
=u_{0}\right\}
\end{equation*}%
is called a solution of the problem (0.1) - (0.2) if the following equation
is fulfilled 
\begin{equation}
\underset{0}{\overset{\infty }{\int }}\ \underset{\Omega }{\int }\left[ i%
\frac{\partial u}{\partial t}-\Delta u+f\left( x,u\right) \right] ~\overline{%
\varphi }\ dxdt=\underset{0}{\overset{\infty }{\int }}\ \underset{\Omega }{%
\int }h\ \overline{\varphi }\ dxdt  \tag{1.3}
\end{equation}%
for any $\varphi \in L^{\left( 2,m\right) }\left( R_{+};W_{0}^{1,2}\left(
\Omega \right) \right) $.
\end{definition}

It should be noted that the sense in which equation (1.3) is to be
understood will be explained below (section 2). We have proved the following
result for the considered problem.

\begin{theorem}
Let the function $f$ have the representation ($1.1$) in the generalized
sense, where $q\in W_{p_{0}}^{-1}\left( \Omega \right) $ is a nonnegative
distribution (generalized function,\footnote{%
see, Definition 2 of the section 2}), $p_{0}=\frac{2n}{2\left( n-1\right)
-p\left( n-2\right) }$, $\frac{2\left( n-1\right) }{n-2}>$ $p>2$ if $n\geq 3$%
; $p_{0},p>2$ are arbitrary if $n=2$, and $p_{0},p\geq 2$ are arbitrary if $%
n=1$ (in particular, if $n=3$ then $2<p<4$ and $p_{0}=\frac{6}{4-p}$) and
conditions (i), (ii) are fulfilled. Then for any $h\in L^{2}\left( Q\right) $
and $u_{0}\in W_{0}^{1,2}\left( \Omega \right) $ the problem (0.1) - (0.2)
is solvable in $L^{\left( 2,m\right) }\left( R_{+};W_{0}^{1,2}\left( \Omega
\right) \right) \cap \left\{ u\left\vert \ \frac{\partial u}{\partial t}\in
L^{2}\left( R_{+};L^{2}\left( \Omega \right) \right) \right. ;\ u\left(
0,x\right) =u_{0}\right\} $.
\end{theorem}

For the investigation of the considered problem we used some general
solvability theorems, which are conducted in section 3. We begin with
explanation of equation (1.3).

\section{The Solution Concept and Function Spaces}

So we will consider the case when the function $f\left( x,u\right) $ has the
form ($1.1$) where functions $a$, $q$ and $u$ are the same as above.
Consequently the function $q\left( x\right) $ is a generalized function,
which has singularity of order 1. Therefore we must understand the equation
(0.1) in the generalized function space sense, i.e. 
\begin{equation*}
\underset{\Omega }{\int }\left[ i\frac{\partial u}{\partial t}-\Delta
u+f\left( x,u\right) \right] ~\overline{\varphi }\left( x\right) dx\equiv
\end{equation*}%
\begin{equation*}
\underset{\Omega }{\int }\left[ i\frac{\partial u}{\partial t}-\Delta
u\left( t,x\right) +q\left( x\right) \left\vert u\left( t,x\right)
\right\vert ^{p-2}u\left( t,x\right) \right] ~\overline{\varphi }\left(
x\right) dx-
\end{equation*}%
\begin{equation}
\underset{\Omega }{\int }a\left( x\right) \left\vert u\left( t,x\right)
\right\vert ^{\widetilde{p}-2}u\left( t,x\right) ~\overline{\varphi }\left(
x\right) dx=\underset{\Omega }{\int }h\left( t,x\right) \overline{\varphi }%
\left( x\right) dx  \tag{2.1}
\end{equation}%
for any $\varphi \equiv \varphi _{1}+i\varphi _{2}$, $\varphi _{j}\in
D\left( \Omega \right) $, $j=1,2$ ,\ where $D\left( \Omega \right) $ is $%
C_{0}^{\infty }\left( \Omega \right) $ and \textrm{supp}$\varphi _{j}\subset
\Omega $ with corresponding topology. Here the equation (2.1) will be
understood in the sense of the space $L_{2}\left( R_{+}\right) $.

In the beginning we need to define the expression $q\ \left\vert
u\right\vert ^{p-2}u$. It is known that (see, for example, [21]) in the case
when $q\in W_{p_{0}}^{-1}\left( \Omega \right) $ we can represent it in the
form $q\left( x\right) \equiv \underset{k=0}{\overset{n}{\sum }}%
D_{k}q_{k}\left( x\right) $, $D_{k}\equiv \frac{\partial }{\partial x_{k}}$, 
$D_{0}\equiv I$, $q_{k}\in L_{p_{0}}\left( \Omega \right) $, $k=0,\overline{%
1,n}$ \ in the generalized function space sense. From here it follows that
if a solution of the considered problem belongs to the space which contains
to $W_{0}^{1,\widetilde{p}_{1}}\left( \Omega \right) $ for some number $%
\widetilde{p}_{1}>1$ then we can understand the term $q\ \left\vert
u\right\vert ^{p-2}u$ in the following sense 
\begin{equation}
\left\langle q\ \left\vert u\right\vert ^{p-2}u,\varphi \right\rangle \equiv 
\underset{\Omega }{\int }q\left( x\right) \left\vert u\left( t,x\right)
\right\vert ^{p-2}u\left( t,x\right) \overline{\varphi }\left( x\right) dx 
\tag{2.2}
\end{equation}%
for any $\varphi \equiv \varphi _{1}+i\varphi _{2}$, $\varphi _{j}\in
D\left( \Omega \right) $, $j=1,2$ and a.e. $t>0$. Therefore we must find the
needed number $\widetilde{p}_{1}\geq 2$. Namely we must find the relation
between the numbers $p_{0}$ and $\widetilde{p}_{1}$. So, taking into account
that for a function $u\in L^{m}\left( R_{+};W_{0}^{1,2}\left( \Omega \right)
\right) $, i.e. $\widetilde{p}_{1}=2$ (as $h\in L^{2}\left( Q\right) $ by
the assumption) we have $u\in L^{m}\left( R_{+};L^{\widetilde{p}_{1}^{\ast
}}\left( \Omega \right) \right) $, where $\widetilde{p}_{1}^{\ast }=2^{\ast
}=\frac{2n}{n-2}$ for $n\geq 3$ by virtue of the embedding theorem, from
(2.2) we get 
\begin{equation*}
\left\langle q\ \left\vert u\right\vert ^{p-2}u,\varphi \right\rangle \equiv 
\underset{\Omega }{\int }q\left( x\right) \left\vert u\left( t,x\right)
\right\vert ^{p-2}u\left( t,x\right) \overline{\varphi }\left( x\right) dx=
\end{equation*}%
\begin{equation*}
\underset{\Omega }{\int }\ \underset{k=0}{\overset{n}{\sum }}\frac{\partial 
}{\partial x_{k}}q_{k}\left( x\right) \left\vert u\left( t,x\right)
\right\vert ^{p-2}u\left( t,x\right) \overline{\varphi }\left( x\right) dx=-%
\underset{\Omega }{\int }\ \underset{k=1}{\overset{n}{\sum }}q_{k}\left\vert
u\right\vert ^{p-2}u\frac{\partial \overline{\varphi }}{\partial x_{k}}dx-
\end{equation*}%
\begin{equation}
\left( p-1\right) \underset{\Omega }{\int }\ \underset{k=1}{\overset{n}{\sum 
}}q_{k}\left\vert u\right\vert ^{p-2}\frac{\partial u}{\partial x_{k}}%
\overline{\varphi }dx+\underset{\Omega }{\int }\ q_{0}\left\vert
u\right\vert ^{p-2}u\varphi dx=I_{1}+I_{2}+\underset{\Omega }{\int }\
q_{0}\left\vert u\right\vert ^{p-2}u\overline{\varphi }dx  \tag{2.3}
\end{equation}%
by virtue of the generalized function theory.

Here and in what follows we assume $n\geq 3$. Because if $n=1,2$ then we can
choose arbitrary $p\geq 2$, as will be observed below. Let us take into
account that $\varphi _{j}\in D\left( \Omega \right) $ and $n\geq 3$, then
in order for the expression in the left part of (2.3) to have the meaning,
it is enough for us to take $1\leq p-1\leq \frac{2n\left( p_{0}-1\right) }{%
p_{0}\left( n-2\right) }$ for the integral $I_{1}$ and $0\leq p-2\leq \frac{%
n\left( p_{0}-2\right) }{p_{0}\left( n-2\right) }$ for the integral $I_{2}$.
Therefore if $2\leq p\leq \frac{3np_{0}-2\left( n+2p_{0}\right) }{%
p_{0}\left( n-2\right) }$ then the left part of (2.3) is defined. Now, let $%
\varphi _{j}\in W_{0}^{1,2}\left( \Omega \right) $, $j=1,2$. Then it is
sufficient to study one of the $I_{1}$ and $I_{2}$. Let us consider $I_{1}$,
from which we obtain, that $2\leq p\leq \frac{2np_{0}-2\left( n+p_{0}\right) 
}{p_{0}\left( n-2\right) }$, moreover we can choose $p\geq 2$ only if $%
p_{0}>n$. On the other hand, if we take into account that given $p,$ we
obtain $p_{0}=\frac{2n}{2\left( n-1\right) -p\left( n-2\right) }$, and
consequently in order for $p_{0}<\infty $ we must choose $2\left( n-1\right)
>p\left( n-2\right) $ or $p<\frac{2\left( n-1\right) }{n-2}$. In the case
when $n=3$ then $p<4$ and $p_{0}=\frac{6}{4-p}$.

Thus we determined under what conditions the left part of (2.3) is defined.
Hence that implies the correctness of the statement

\begin{proposition}
Assume $\widetilde{f}$ be an operator defined by expression $\widetilde{f}%
\left( u\right) \equiv q\ \left\vert u\right\vert ^{p-2}u$, where $q\in
W^{-1,p_{0}}\left( \Omega \right) $, and $u\in L^{\left( 2,m\right) }\left(
R_{+};W_{0}^{1,2}\left( \Omega \right) \right) $. If $2\leq p<\frac{2\left(
n-1\right) }{n-2}$ and $p_{0}=\frac{2n}{2\left( n-1\right) -p\left(
n-2\right) }$\ if $n\geq 3$ (in particular, if $n=3$ then $2\leq p<4$ and $%
p_{0}=\frac{6}{4-p}$) then $\widetilde{f}$ :$L^{\left( 2,m\right) }\left(
R_{+};W_{0}^{1,2}\left( \Omega \right) \right) \longrightarrow L^{2}\left(
R_{+};W^{-1,2}\left( \Omega \right) \right) $ is a bounded operator.
\end{proposition}

So that exactly explains Proposition 1 and the representation (2.2) for any $%
\varphi _{j}\in W_{0}^{1,2}\left( \Omega \right) $ we consider the following
class of the functions $u:\Omega \longrightarrow C$ 
\begin{equation}
\emph{M}_{\eta ,W^{1,\beta }\left( \Omega \right) }\equiv \left\{ u\in
L\left( \Omega \right) ~\left\vert ~\eta \left( u\right) \in W^{1,\beta
}\left( \Omega \right) ,\ \eta \left( u\right) \equiv \left\vert
u\right\vert ^{\frac{\alpha }{\beta }}u\right. \right\} \equiv S_{1,\alpha
,\beta }\left( \Omega \right)  \tag{2.4}
\end{equation}%
where $\alpha \geq 0$, $\beta >1$ are certain numbers, $W^{1,\beta }\left(
\Omega \right) $ is a Sobolev space, i.e. we consider a class of the $pn$%
-spaces\footnote{%
These are a complete metric spaces; about their properties see, for example,
\par
Soltanov K. N. , Some nonlinear equations of the nonstable filtration type
and embedding theorems. J. Nonlinear Analysis : T.M. \& APPL. (2006), 65,
2103-2134 and references therein}.

It is not difficult to see that if $1\leq \alpha _{0}+\beta _{0}\leq \alpha
_{1}+\beta _{1}$, $0\leq \beta _{0}<\beta _{1}$, $\alpha _{1}\beta _{0}\leq
\alpha _{0}\beta _{1}$, $1\leq \beta _{1}$ then 
\begin{equation}
\underset{\Omega }{\int }\left\vert u\right\vert ^{\alpha _{0}}\underset{k=1}%
{\overset{n}{\sum }}\left\vert D_{k}u\right\vert ^{\beta _{0}}dx\leq c%
\underset{\Omega }{\int }\left\vert u\right\vert ^{\alpha _{1}}\underset{k=1}%
{\overset{n}{\sum }}\left\vert D_{k}u\right\vert ^{\beta _{1}}dx+c_{1} 
\tag{2.5}
\end{equation}%
holds for any $u\in C_{0}^{1}\left( \Omega \right) $, where\ constants $%
c,c_{1}\geq 0$ are independent from $u$.

Furthermore if we will introduce the space $\emph{M}_{\eta ,W_{0}^{1,\beta
}\left( \Omega \right) }\equiv \overset{0}{S}_{1,\alpha ,\beta }\left(
\Omega \right) \equiv S_{1,\alpha ,\beta }\left( \Omega \right) \cap \left\{
u\left( x\right) \left\vert \ u\left\vert \ _{\partial \Omega }\right.
=0\right. \right\} $ then we get

\begin{lemma}
Let $u\in W_{0}^{1,2}\left( \Omega \right) $ and the number $p$ satisfy the
inequation $2<p<\frac{2\left( n-1\right) }{n-2}$, $n\geq 3$. Then the
function $v\left( x\right) \equiv \eta \left( u\left( x\right) \right)
\equiv \left\vert u\left( x\right) \right\vert ^{p}$ belongs to $%
W_{0}^{1,\beta }\left( \Omega \right) $ for any $\beta \in \left[
1,p_{0}^{\prime }\right] $, where $p_{0}=\frac{2n}{2\left( n-1\right)
-p\left( n-2\right) }$ and $p_{0}^{\prime }=\frac{p_{0}}{p_{0}-1}=\frac{2n}{%
p\left( n-2\right) +2}$. (It is obvious: $u\in W_{0}^{1,2}\left( \Omega
\right) \Longrightarrow v\equiv \left\vert u\right\vert ^{p}\in
W_{0}^{1,\beta }\left( \Omega \right) $ for any $\beta \in \left[ 1,2\right) 
$ if $n=2$, and for any $\beta \in \left[ 1,2\right] $ if $n=1$.)
\end{lemma}

\begin{proof}
We have 
\begin{equation*}
\underset{\Omega }{\int }~\left\vert u\right\vert ^{\left( p-1\right) \beta
}\left\vert D_{k}u\right\vert ^{\beta }dx\leq k\left( \varepsilon \right) 
\underset{\Omega }{\int }~\left\vert D_{k}u\right\vert ^{2}dx+\varepsilon 
\underset{\Omega }{\int }~\left\vert u\right\vert ^{\left( p-1\right) \beta 
\frac{2}{2-\beta }}dx
\end{equation*}%
for any $u\in W\ ^{1,2}\left( \Omega \right) $ and $\beta \in \left[
1,p_{0}^{\prime }\right] $. It is enough to consider the case $\beta
=p_{0}^{\prime }=\frac{2n}{p\left( n-2\right) +2}$, because $\Omega \subset
R^{n}$ is a bounded domain with sufficiently smooth boundary $\partial
\Omega $. So, from here we get 
\begin{equation*}
\underset{\Omega }{\int }~\left\vert u\right\vert ^{\left( p-1\right) \beta
}\left\vert D_{k}u\right\vert ^{\beta }dx\leq c\left( \varepsilon \right) 
\underset{\Omega }{\int }~\left\vert D_{k}u\right\vert ^{2}dx+\varepsilon 
\underset{\Omega }{\int }~\left\vert u\right\vert ^{\left( p-1\right)
p_{0}^{\prime }\frac{2}{2-p_{0}^{\prime }}}dx+c_{0}.
\end{equation*}

Then $\left( p-1\right) p_{0}^{\prime }\frac{2}{2-p_{0}^{\prime }}=\frac{2n}{%
n-2}$ holds under the conditions of the Lemma 1. Consequently, we obtain $%
u\in W_{0}^{1,2}\left( \Omega \right) \Longrightarrow v\equiv \left\vert
u\right\vert ^{p}\in W_{0}^{1,\beta }\left( \Omega \right) $ for any $\beta
\in \left[ 1,p_{0}^{\prime }\right] $, with choosing $\varepsilon >0$
sufficiently small and with use the Embedding Theorem for Sobolev spaces.
\end{proof}

\begin{corollary}
Let $u,w\in W_{0}^{1,2}\left( \Omega \right) $ and the number $p$ is such
that $2<p<\frac{2\left( n-1\right) }{n-2}$, $n\geq 3$. Then the function $%
v\left( x\right) \equiv \left\vert u\left( x\right) \right\vert
^{p-2}u\left( x\right) w\left( x\right) $ belongs to $W_{0}^{1,\beta }\left(
\Omega \right) $ (i.e. $v\in W_{0}^{1,\beta }\left( \Omega \right) $) for
any $\beta \in \left[ 1,p_{0}^{\prime }\right] $, where $p_{0}=\frac{2n}{%
2\left( n-1\right) -p\left( n-2\right) }$ and $p_{0}^{\prime }=\frac{p_{0}}{%
p_{0}-1}$.
\end{corollary}

Now we introduce a concept of the nonnegative generalized function

\begin{definition}
A generalized function $q\left( x\right) $ is called a non-negative
distribution (\textquotedblleft $q\geq 0$\textquotedblright ) iff $%
\left\langle q,\ \varphi \right\rangle \geq 0$\ holds for any non-negative
test function $\varphi \in D\left( \Omega \right) $.
\end{definition}

\section{General Solvability Results}

Let $X,Y$ be reflexive Banach spaces and $X^{\ast },Y^{\ast }$ their dual
spaces, moreover $Y$ is a reflexive Banach space with strictly convex norm
together with $Y^{\ast }$ (see, for example, references of [29]). Let $%
f:D\left( f\right) \subseteq X\longrightarrow Y$ be an operator. So we
conduct variant of the main result of [29] (the more general cases can be
seen in [30]). Consider the following conditions:

\textit{(a)} $X,Y$ be Banach spaces such as above and $f:D\left( f\right)
\subseteq X\longrightarrow Y$ be a continuous mapping, moreover there is the
closed ball $B_{r_{0}}^{X}\left( x_{0}\right) \subset X$ of an element $%
x_{0} $ of $D\left( f\right) $ that belongs to $D\left( f\right) $ ($%
B_{r_{0}}\left( x_{0}\right) $ $\subseteq D\left( f\right) $)\footnote{%
Here it is enough assume: there is the closed neighborhood $U_{\delta
}\left( x_{0}\right) $ $\subset X$ of an element $x_{0}$ of $D\left(
f\right) $ that belongs to $D\left( f\right) $ ($U_{\delta }\left(
x_{0}\right) $ $\subseteq D\left( f\right) $) and $U_{\delta }\left(
x_{0}\right) $ is equivalence to $B_{r_{0}}^{X}\left( x_{0}\right) $ for
some numbers $\delta ,r_{0}>0$. Consequently, it is enough account that $%
U_{r_{0}}\left( x_{0}\right) \equiv B_{r_{0}}^{X}\left( x_{0}\right) $.};

Let the following conditions are fulfilled on the closed ball $%
B_{r_{0}}^{X}\left( x_{0}\right) $ $\subseteq D\left( f\right) $ : $\ $

\textit{(b) }$f$ is a bounded mapping on the ball $B_{r_{0}}^{X}\left(
x_{0}\right) $, i.e. $\left\Vert f\left( x\right) \right\Vert _{Y}\leq \mu
\left( \left\Vert x\right\Vert _{X}\right) $ holds for $\forall x\in
B_{r_{0}}^{X}\left( x_{0}\right) $ where $\mu :R_{+}^{1}\longrightarrow
R_{+}^{1}$ is a continuous function;

\textit{(c)} there is a mapping $g:D\left( g\right) \subseteq
X\longrightarrow Y^{\ast }$, and a continuous function $\nu
:R_{+}^{1}\longrightarrow R^{1}$ nondecreasing for $\tau \geq \tau _{0}$
such that $D\left( f\right) \subseteq D\left( g\right) $, and for any $%
S_{r}^{X}\left( x_{0}\right) \subset B_{r_{0}}^{X}\left( x_{0}\right) $, $%
0<r\leq r_{0}$, closure of\textrm{\ }$g\left( S_{r}^{X}\left( x_{0}\right)
\right) \equiv S_{r}^{Y^{\ast }}\left( 0\right) $, $S_{r}^{X}\left(
x_{0}\right) \subseteq g^{-1}\left( S_{r}^{Y^{\ast }}\left( 0\right) \right) 
$ 
\begin{equation}
\left\langle f\left( x\right) -f\left( x_{0}\right) ,g\left( x\right)
\right\rangle \geq \nu \left( \left\Vert x-x_{0}\right\Vert _{X}\right)
\left\Vert x-x_{0}\right\Vert _{X},\text{ }  \tag{3.1}
\end{equation}%
\begin{equation*}
\text{a.e. }x\in B_{r_{0}}^{X}\left( x_{0}\right) \quad \&\ \nu \left(
r_{0}\right) \geq \delta _{0}>0
\end{equation*}%
holds, here $\delta _{0}>0$, $\tau _{0}\geq 0$ are constants;

\textit{(d)} almost each $\widetilde{x}\in intB_{r_{0}}^{X}\left(
x_{0}\right) $ possesses a neighborhood $V_{\varepsilon }\left( \widetilde{x}%
\right) $, $\varepsilon \geq \varepsilon _{0}>0$ such that the inequation 
\begin{equation}
\left\Vert f\left( x_{2}\right) -f\left( x_{1}\right) \right\Vert _{Y}\geq
\Phi \left( \left\Vert x_{2}-x_{1}\right\Vert _{X},\widetilde{x},\varepsilon
\right) +\psi \left( \left\Vert x_{1}-x_{2}\right\Vert _{Z},\widetilde{x}%
,\varepsilon \right)  \tag{3.2}
\end{equation}%
holds for any $x_{1},x_{2}\in V_{\varepsilon }\left( \widetilde{x}\right)
\cap B_{r_{0}}^{X}\left( x_{0}\right) $, where $\Phi \left( \tau ,\widetilde{%
x},\varepsilon \right) \geq 0$ is a continuous function of $\tau $ and $\Phi
\left( \tau ,\widetilde{x},\varepsilon \right) =0\Leftrightarrow \tau =0$
(in particular, maybe $\widetilde{x}=0$, $\varepsilon =\varepsilon
_{0}=r_{0} $ and $V_{\varepsilon }\left( \widetilde{x}\right)
=V_{r_{0}}\left( x_{0}\right) \equiv B_{r_{0}}^{X}\left( x_{0}\right) $,
consequently$\ \Phi \left( \tau ,\widetilde{x},\varepsilon \right) \equiv
\Phi \left( \tau ,x_{0},r_{0}\right) $ on $B_{r_{0}}^{X}\left( x_{0}\right) $%
), $Z$ is a Banach space and the inclusion $X\subset Z$ is compact, and $%
\psi \left( \cdot ,\widetilde{x},\varepsilon \right)
:R_{+}^{1}\longrightarrow R^{1}$ is a continuous function at $\tau $ and $%
\psi \left( 0,\widetilde{x},\varepsilon \right) =0$;

\textit{(d') } $f$ possesses the \textrm{P-property} on the ball $%
B_{r_{0}}^{X}\left( x_{0}\right) $, i.e. for any precompact subset $%
M\subseteq \func{Im}\ f$ of $Y$ there exists a (general) subsequence $%
M_{0}\subset M$ such that there exists a precompact subset $G$ of $%
B_{r_{0}}^{X}\left( x_{0}\right) \subset X$ that satisfies the inclusions $%
f^{-1}\left( M_{0}\right) \subseteq G$ and $f\left( G\cap D\left( f\right)
\right) \supseteq M_{0}$.

\begin{theorem}
Let the conditions (a), \textit{(b),} \textit{(c)} be fulfilled. Then if the
image $f\left( B_{r_{0}}^{X}\left( x_{0}\right) \right) $ of the ball $%
B_{r_{0}}^{X}\left( x_{0}\right) $ is closed (or is fulfilled the condition
(d) or (d')), then $f\left( B_{r_{0}}^{X}\left( x_{0}\right) \right) $ is a
bodily subset (i.e. with nonempty interior) of $Y$, moreover $f\left(
B_{r_{0}}^{X}\left( x_{0}\right) \right) $ contains a bodily subset $M$ that
has the form 
\begin{equation*}
M\equiv \left\{ y\in Y\left\vert \ \left\langle y,g\left( x\right)
\right\rangle \leq \left\langle f\left( x\right) ,g\left( x\right)
\right\rangle ,\right. \forall x\in S_{r_{0}}^{X}\left( x_{0}\right)
\right\} .
\end{equation*}
\end{theorem}

Now we lead a solvability theorem for the nonlinear equation in Banach
spaces, which is proved using Theorem 2. Let $F_{0}:D\left( F\right)
\subseteq X$ $\longrightarrow Y$ and $F_{1}:D\left( F_{1}\right) \subseteq X$
$\longrightarrow Y$ be some nonlinear mappings such that $D\left(
F_{0}\right) \cap D\left( F_{1}\right) =G\subseteq X$ and $G\neq \varnothing 
$. Consider the following equation 
\begin{equation}
F\left( x\right) \equiv F_{0}\left( x\right) +F_{1}\left( x\right) =y,\quad
y\in Y  \tag{3.3}
\end{equation}%
where $y$ is an arbitrary element of $Y$.

Let $B_{r}^{X}\left( x_{0}\right) \subseteq D\left( F_{0}\right) \cap
D\left( F_{1}\right) \subseteq X$ be the closed ball, $r>0$ be a number.
Consider the following conditions:

1) $F_{0}:B_{r}^{X}\left( x_{0}\right) $ $\longrightarrow Y$ is a bounded
continuous operator together with its inverse operator $F_{0}^{-1}$, (as $%
F_{0}^{-1}:$ $D\left( F_{0}^{-1}\right) \subseteq $ $Y$ $\longrightarrow $ $%
X $);

2) $F_{1}:B_{r}^{X}\left( x_{0}\right) $ $\longrightarrow Y$ is a nonlinear
continuous operator;

3) There are continuous functions $\mu _{i}:R_{+}^{1}\longrightarrow
R_{+}^{1}$ , $i=1,2$ and $\nu :R_{+}^{1}\longrightarrow R^{1}$ such that the
inequations 
\begin{equation*}
\left\Vert F_{0}\left( x\right) -F_{0}\left( x_{0}\right) \right\Vert
_{Y}\leq \mu _{1}\left( \left\Vert x-x_{0}\right\Vert _{X}\right) \ \&\
\left\Vert F_{1}\left( x\right) -F_{1}\left( x_{0}\right) \right\Vert
_{Y}\leq \mu _{2}\left( \left\Vert x-x_{0}\right\Vert _{X}\right) ,
\end{equation*}%
\begin{equation*}
\left\langle F\left( x\right) -F\left( x_{0}\right) ,g\left( x\right)
\right\rangle \geq c\left\langle F_{0}\left( x\right) -F_{0}\left(
x_{0}\right) ,g\left( x\right) \right\rangle \geq \nu \left( \left\Vert
x-x_{0}\right\Vert _{X}\right) \left\Vert x-x_{0}\right\Vert _{X}
\end{equation*}%
hold for any $x\in B_{r}^{X}\left( x_{0}\right) $, moreover $\nu \left(
r\right) \geq $ $\delta _{0}$ holds for some number $\delta _{0}>0$, where
the mapping $g:B_{r}^{X}\left( x_{0}\right) \subseteq D\left( g\right)
\subseteq X\longrightarrow Y^{\ast }$ fulfills the conditions of Theorem 2, $%
c>0$ is some number.

4) Almost each $\widetilde{x}\in intB_{r}^{X}\left( x_{0}\right) $ possesses
a neighborhood $B_{\varepsilon }^{X}\left( \widetilde{x}\right) $, $%
\varepsilon \geq \varepsilon _{0}>0$, such that the inequation 
\begin{equation*}
\left\Vert F\left( x_{1}\right) -F\left( x_{2}\right) \right\Vert _{Y}\geq
c_{1}\left\Vert F_{0}\left( x_{1}\right) -F_{0}\left( x_{2}\right)
\right\Vert _{Y}\geq
\end{equation*}%
\begin{equation*}
k_{0}\left( \left\Vert x_{1}-x_{2}\right\Vert _{X},\widetilde{x},\varepsilon
\right) -k_{1}\left( \left\Vert x_{1}-x_{2}\right\Vert _{Z},\widetilde{x}%
,\varepsilon \right) ,\quad X\Subset Z
\end{equation*}%
holds for any $x_{1},x_{2}\in B_{\varepsilon }^{X}\left( \widetilde{x}%
\right) $ and some number $\varepsilon _{0}>0$, where $k_{i}\left( \tau ,%
\widetilde{x},\varepsilon \right) \geq 0,$ $i=0,1$ are continuous functions
of $\tau $ for any given $\widetilde{x}$, and such that $k_{0}\left( \tau ,%
\widetilde{x},\varepsilon \right) =0\Longleftrightarrow \tau =0$, $%
k_{1}\left( 0,\widetilde{x},\varepsilon \right) =0$,\ and $X\Subset Z$ (i.e. 
$X\subset Z$ is compact).

Then the following statement is true, which follows from Theorem 2.

\begin{theorem}
Let the conditions 1, 2, 3 be fulfilled. Then if $F\left( B_{r}^{X}\left(
x_{0}\right) \right) $ is closed (or is fulfilled the condition 4 or (d')),
then the equation (3.3) has a solution in the ball $B_{r}^{X}\left(
x_{0}\right) $ for any $y\in Y$ satisfying the inequation 
\begin{equation*}
\left\langle y-F\left( x_{0}\right) ,g\left( x\right) \right\rangle \leq \nu
\left( \left\Vert x-x_{0}\right\Vert _{X}\right) \left\Vert
x-x_{0}\right\Vert _{X},\quad \forall x\in S_{r}^{X}\left( x_{0}\right) .
\end{equation*}
\end{theorem}

\section{Proof of Existence Theorem of Problem (0.1)-(0.2)}

It should be noted that here we continue the investigation of the problem
studied in the article [31] where this problem is studied in the case $%
Q\equiv Q_{T}\equiv \left( 0,T\right) \times \Omega $, when $T<\infty $ is
some number. Here we will study the global existence of this problem, i.e.
in the case when $Q\equiv R_{+}\times \Omega $. So in the beginning we set a
space and explain the way of the investigation. As the solutions $u\left(
t,x\right) $ of the considered problem will seek in the form $u\left(
t,x\right) \equiv u_{1}\left( t,x\right) +iu_{2}\left( t,x\right) $, where $%
u_{j}:Q\longrightarrow R$, $j=1,2$ we can set this function $u\left(
t,x\right) $ as the vector function, i.e. $\overrightarrow{u\left(
t,x\right) }\equiv \left( u_{1}\left( t,x\right) ;u_{2}\left( t,x\right)
\right) $ and $\overrightarrow{u}:Q\longrightarrow R^{2}$. Consequently if
we write $u\in X$ (for example $X\equiv L^{m}\left( R_{+};W_{0}^{1,2}\left(
\Omega \right) \right) \cap W^{1,2}\left( R_{+};L^{2}\left( \Omega \right)
\right) $, $m\geq \max \left\{ p,\widetilde{p}\right\} $) then this we
understand as $u_{j}\in X$, $j=1,2$ or $\overrightarrow{u}\in X\times X$.
Now we can define a solution of the problem (0.1)-(0.2) more exactly.

\begin{definition}
We say that the function $u\in L^{m}\left( R_{+};W_{0}^{1,2}\left( \Omega
\right) \right) \cap W^{1,2}\left( R_{+};L^{2}\left( \Omega \right) \right) $
$\cap \left\{ w\left( t,x\right) \left\vert \ w\left( 0,x\right)
=u_{0}\left( x\right) \right. \right\} \equiv X$ (as complex function) is a
solution of the problem (0.1)-(0.2) if it satisfies the equation 
\begin{equation*}
i\left\langle \frac{\partial u}{\partial t},\overline{v}\right\rangle
+\left\langle \nabla u,\nabla \overline{v}\right\rangle +\left\langle
q\left( x\right) \left\vert u\right\vert ^{p-2}u,\overline{v}\right\rangle
+\left\langle a\left( x\right) \left\vert u\right\vert ^{\widetilde{p}-2}u,%
\overline{v}\right\rangle =\left\langle h,\overline{v}\right\rangle
\end{equation*}%
for any $v\in L^{\left( 2,m\right) }\left( R_{+};W_{0}^{1,2}\left( \Omega
\right) \right) $ and a.e. $t>0$.
\end{definition}

Let us $f:X\longrightarrow Y$ is the operator generated by the problem
(0.1)-(0.2), where $X$ is the denoted above space, and 
\begin{equation*}
Y\equiv L^{2}\left( R_{+};W^{-1,2}\left( \Omega \right) \right)
+L^{m^{\prime }}\left( R_{+};W^{-1,2}\left( \Omega \right) \right) +L^{%
\widetilde{q}}\left( Q\right)
\end{equation*}%
where $\widetilde{q}=\frac{\widetilde{p}}{\widetilde{p}-1}$. We will show
that for this operator are fulfilled all conditions of the main theorem. We
make this by sequence of steps. Clearly that the conditions (a) and (b)
fulfilled. Indeed, the explanations conducted in the previous sections shows
that $f:X\longrightarrow Y$ is the continuous bounded operator. It should be
noted that the calculation of the function $\mu $ not is difficult,
therefore we not will conduct this computation here (see, below Proposition
4).

\begin{proposition}
Let all conditions of Theorem 1 are fulfilled, then the operator $f$
satisfies the condition (c) with the operator $\frac{\partial }{\partial t}%
+I $ on the space $W^{1,2}\left( R_{+};W_{0}^{1,2}\left( \Omega \right)
\right) \cap $ $L^{m}\left( R_{+};W_{0}^{1,2}\left( \Omega \right) \right)
\cap $ $\left\{ v\left( t,x\right) \left\vert \ v\left( 0,x\right)
=u_{0}\left( x\right) \right. \right\} $. Moreover takes place the following
inequations 
\begin{equation*}
\overset{t}{\underset{0}{\dint }}\func{Im}\left\langle f\left( u\right) ,%
\frac{\partial \overline{u}}{\partial s}+\overline{u}\right\rangle ds\equiv 
\frac{1}{2}\left\Vert u\left( t\right) \right\Vert _{2}^{2}-\frac{1}{2}%
\left\Vert u_{0}\right\Vert _{2}^{2}+\overset{t}{\underset{0}{\dint }}%
\left\Vert \frac{\partial u}{\partial s}\right\Vert _{L_{2}\left( \Omega
\right) }^{2}ds;
\end{equation*}%
\begin{equation*}
\overset{t}{\underset{0}{\dint }}\func{Re}\left\langle f\left( u\right) ,%
\frac{\partial \overline{u}}{\partial s}+\overline{u}\right\rangle ds\geq
\eta \left\Vert \nabla u\left( t\right) \right\Vert _{2}^{2}+\eta \overset{t}%
{\underset{0}{\dint }}\left\Vert \nabla u\left( s\right) \right\Vert
_{2}^{2}ds+
\end{equation*}%
\begin{equation*}
\delta _{1}\overset{t}{\underset{0}{\dint }}\left\langle q\left( x\right)
\left\vert u\right\vert ^{p-2}u,\overline{u}\right\rangle \left( s\right)
ds+\delta _{2}\left\langle q\left( x\right) \left\vert u\right\vert ^{p-2}u,%
\overline{u}\right\rangle \left( t\right) -
\end{equation*}%
\begin{equation*}
C\left( \left\Vert \nabla u_{0}\right\Vert ,\left\Vert q\right\Vert
_{W^{-1,2}},\left\Vert u_{0}\right\Vert _{2^{\ast }},\left\Vert a\right\Vert
_{m},p,\widetilde{p}\right) ,
\end{equation*}%
the constants of these inequations are determined in (4.6)
\end{proposition}

\begin{proof}
Consider the expression $\left\langle f\left( u\right) ,\frac{\partial 
\overline{u}}{\partial t}+\overline{u}\right\rangle $, which we can write as 
\begin{equation*}
\left\langle f\left( u\right) ,\frac{d\overline{u}}{dt}\right\rangle
=i\left\langle \frac{\partial u}{\partial t},\frac{\partial \overline{u}}{%
\partial t}\right\rangle +\frac{1}{2}\frac{d}{dt}\left\langle \nabla
u,\nabla \overline{u}\right\rangle +
\end{equation*}%
\begin{equation}
\left\langle q\left( x\right) \left\vert u\right\vert ^{p-2}u,\frac{\partial 
\overline{u}}{\partial t}\right\rangle +\left\langle a\left( x\right)
\left\vert u\right\vert ^{\widetilde{p}-2}u,\frac{d\overline{u}}{dt}%
\right\rangle ,  \tag{4.1}
\end{equation}%
\begin{equation}
\left\langle f\left( u\right) ,\overline{u}\right\rangle =i\left\langle 
\frac{\partial u}{\partial t},\overline{u}\right\rangle +\left\langle \nabla
u,\nabla \overline{u}\right\rangle +\left\langle q\left( x\right) \left\vert
u\right\vert ^{p-2}u,\overline{u}\right\rangle +\left\langle a\left(
x\right) \left\vert u\right\vert ^{\widetilde{p}-2}u,\overline{u}%
\right\rangle ,  \tag{4.2}
\end{equation}%
for any $u\in W^{1,2}\left( R_{+};W_{0}^{1,2}\left( \Omega \right) \right) $%
. We begin by (4.2), then we get 
\begin{equation*}
\left\langle f\left( u\right) ,\overline{u}\right\rangle =\frac{i}{2}\frac{d%
}{dt}\left\Vert u\left( t\right) \right\Vert _{2}^{2}+\left\Vert \nabla
u\right\Vert _{2}^{2}+
\end{equation*}%
\begin{equation*}
\left\langle q\left( x\right) \left\vert u\right\vert ^{p-2}u,\overline{u}%
\right\rangle +\left\langle a\left( x\right) \left\vert u\right\vert ^{%
\widetilde{p}-2}u,\overline{u}\right\rangle .
\end{equation*}%
Consequently, we have 
\begin{equation*}
\func{Im}\left\langle f\left( u\right) ,\overline{u}\right\rangle =\frac{1}{2%
}\frac{d}{dt}\left\Vert u\left( t\right) \right\Vert _{2}^{2}
\end{equation*}%
and 
\begin{equation*}
\func{Re}\left\langle f\left( u\right) ,\overline{u}\right\rangle
=\left\Vert \nabla u\right\Vert _{2}^{2}+\left\langle q\left( x\right)
\left\vert u\right\vert ^{p-2}u,\overline{u}\right\rangle +\left\langle
a\left( x\right) \left\vert u\right\vert ^{\widetilde{p}-2}u,\overline{u}%
\right\rangle ,
\end{equation*}%
as $\func{Im}q\left( x\right) =0$ and $\func{Im}a\left( x\right) =0$. Whence
follows that 
\begin{equation*}
\overset{t}{\underset{0}{\dint }}\func{Im}\left\langle f\left( u\right) ,%
\overline{u}\right\rangle ds=\frac{1}{2}\left\Vert u\left( t\right)
\right\Vert _{2}^{2}-\frac{1}{2}\left\Vert u_{0}\right\Vert _{2}^{2}
\end{equation*}%
and 
\begin{equation*}
\func{Re}\left\langle f\left( u\right) ,\overline{u}\right\rangle \geq
\left\Vert \nabla u\right\Vert _{2}^{2}+\left\langle q\left( x\right)
\left\vert u\right\vert ^{p-2}u,\overline{u}\right\rangle +\left\langle
a\left( x\right) \left\vert u\right\vert ^{\widetilde{p}-2}u,\overline{u}%
\right\rangle \Longrightarrow
\end{equation*}%
\begin{equation*}
\func{Re}\left\langle f\left( u\right) ,\overline{u}\right\rangle \geq
\left\Vert \nabla u\right\Vert _{2}^{2}+\left( 1-k_{1}\right) \left\langle
q\left( x\right) ,\left\vert u\right\vert ^{p}\right\rangle -k_{0}\left\Vert
u\right\Vert _{p_{2}}^{2}\Longrightarrow
\end{equation*}%
\begin{equation*}
\func{Re}\left\langle f\left( u\right) ,\overline{u}\right\rangle \geq
\left( 1-k_{0}C^{2}\left( p_{2}\right) \right) \left\Vert \nabla u\left(
t\right) \right\Vert _{2}^{2},
\end{equation*}%
for a.e. $t>0$, as $k_{1}\leq 1$ and $k_{0}C^{2}\left( p_{2}\right) <1$ by
virtue of the condition (\textit{ii}).

Thereby if we examine now (4.1) and (4.2) together then we will get. 
\begin{equation*}
\left\langle f\left( u\right) ,\frac{\partial \overline{u}}{\partial t}+%
\overline{u}\right\rangle \equiv \frac{i}{2}\frac{d}{dt}\left\Vert u\left(
t\right) \right\Vert _{2}^{2}+i\left\Vert \frac{\partial u}{\partial t}%
\right\Vert _{L_{2}\left( \Omega \right) }^{2}+
\end{equation*}%
\begin{equation*}
\left\Vert \nabla u\right\Vert _{2}^{2}+\left\langle q\left( x\right)
\left\vert u\right\vert ^{p-2}u,\overline{u}\right\rangle +\left\langle
a\left( x\right) \left\vert u\right\vert ^{\widetilde{p}-2}u,\overline{u}%
\right\rangle +
\end{equation*}%
\begin{equation*}
\frac{1}{2}\frac{\partial }{\partial t}\left\Vert \nabla u\left( t\right)
\right\Vert _{L_{2}\left( \Omega \right) }^{2}+\frac{1}{p}\frac{\partial }{%
\partial t}\left\langle q\left( x\right) ,\left\vert u\right\vert
^{p}\right\rangle +\frac{1}{\widetilde{p}}\frac{\partial }{\partial t}%
\left\langle a,\left\vert u\right\vert ^{\widetilde{p}}\right\rangle ,
\end{equation*}%
in other words we have 
\begin{equation}
\func{Im}\left\langle f\left( u\right) ,\frac{\partial \overline{u}}{%
\partial t}+\overline{u}\right\rangle \equiv \frac{1}{2}\frac{d}{dt}%
\left\Vert u\left( t\right) \right\Vert _{2}^{2}+\left\Vert \frac{\partial u%
}{\partial t}\right\Vert _{L_{2}\left( \Omega \right) }^{2}  \tag{4.3}
\end{equation}%
and 
\begin{equation*}
\func{Re}\left\langle f\left( u\right) ,\frac{\partial \overline{u}}{%
\partial t}+\overline{u}\right\rangle \equiv \frac{1}{2}\frac{\partial }{%
\partial t}\left\Vert \nabla u\left( t\right) \right\Vert _{L_{2}\left(
\Omega \right) }^{2}+\left\Vert \nabla u\left( t\right) \right\Vert _{2}^{2}+
\end{equation*}%
\begin{equation*}
\frac{1}{p}\frac{\partial }{\partial t}\left\langle q\left( x\right)
,\left\vert u\right\vert ^{p}\right\rangle +\left\langle q\left( x\right)
\left\vert u\right\vert ^{p-2}u,\overline{u}\right\rangle +
\end{equation*}%
\begin{equation}
\frac{1}{\widetilde{p}}\frac{\partial }{\partial t}\left\langle a,\left\vert
u\right\vert ^{\widetilde{p}}\right\rangle +\left\langle a\left( x\right)
\left\vert u\right\vert ^{\widetilde{p}-2}u,\overline{u}\right\rangle . 
\tag{4.4}
\end{equation}

If we integrate with respect to $t$ these equation then we have 
\begin{equation*}
\overset{t}{\underset{0}{\dint }}\func{Im}\left\langle f\left( u\right) ,%
\frac{\partial \overline{u}}{\partial s}+\overline{u}\right\rangle ds\equiv 
\overset{t}{\underset{0}{\dint }}\left[ \frac{1}{2}\frac{d}{ds}\left\Vert
u\left( s\right) \right\Vert _{2}^{2}+\left\Vert \frac{\partial u}{\partial s%
}\right\Vert _{L_{2}\left( \Omega \right) }^{2}\right] ds
\end{equation*}%
and 
\begin{equation*}
\overset{t}{\underset{0}{\dint }}\func{Re}\left\langle f\left( u\right) ,%
\frac{\partial \overline{u}}{\partial s}+\overline{u}\right\rangle ds\equiv 
\overset{t}{\underset{0}{\dint }}\left[ \frac{1}{2}\frac{\partial }{\partial
s}\left\Vert \nabla u\left( s\right) \right\Vert _{2}^{2}+\left\Vert \nabla
u\left( s\right) \right\Vert _{2}^{2}\right] ds+
\end{equation*}%
\begin{equation*}
\overset{t}{\underset{0}{\dint }}\left[ \frac{1}{p}\frac{\partial }{\partial
s}\left\langle q\left( x\right) ,\left\vert u\right\vert ^{p}\right\rangle
+\left\langle q\left( x\right) \left\vert u\right\vert ^{p-2}u,\overline{u}%
\right\rangle \right] ds+
\end{equation*}%
\begin{equation*}
\overset{t}{\underset{0}{\dint }}\left[ \frac{1}{\widetilde{p}}\frac{%
\partial }{\partial s}\left\langle a,\left\vert u\right\vert ^{\widetilde{p}%
}\right\rangle +\left\langle a\left( x\right) \left\vert u\right\vert ^{%
\widetilde{p}-2}u,\overline{u}\right\rangle \right] ds.
\end{equation*}

Thence follow 
\begin{equation}
\overset{t}{\underset{0}{\dint }}\func{Im}\left\langle f\left( u\right) ,%
\frac{\partial \overline{u}}{\partial s}+\overline{u}\right\rangle ds\equiv 
\frac{1}{2}\left\Vert u\left( t\right) \right\Vert _{2}^{2}-\frac{1}{2}%
\left\Vert u_{0}\right\Vert _{2}^{2}+\overset{t}{\underset{0}{\dint }}%
\left\Vert \frac{\partial u}{\partial s}\right\Vert _{L_{2}\left( \Omega
\right) }^{2}ds;  \tag{4.5}
\end{equation}%
\begin{equation*}
\overset{t}{\underset{0}{\dint }}\func{Re}\left\langle f\left( u\right) ,%
\frac{\partial \overline{u}}{\partial s}+\overline{u}\right\rangle ds\equiv 
\frac{1}{2}\left\Vert \nabla u\left( t\right) \right\Vert _{2}^{2}-\frac{1}{2%
}\left\Vert \nabla u_{0}\right\Vert _{2}^{2}+
\end{equation*}%
\begin{equation*}
\overset{t}{\underset{0}{\dint }}\left\Vert \nabla u\left( s\right)
\right\Vert _{2}^{2}ds+\overset{t}{\underset{0}{\dint }}\left[ \left\langle
q\left( x\right) \left\vert u\right\vert ^{p-2}u,\overline{u}\right\rangle
+\left\langle a\left( x\right) \left\vert u\right\vert ^{\widetilde{p}-2}u,%
\overline{u}\right\rangle \right] ds+
\end{equation*}%
\begin{equation*}
\left[ \frac{1}{p}\left\langle q\left( x\right) ,\left\vert u\right\vert
^{p}\right\rangle +\frac{1}{\widetilde{p}}\left\langle a\left( x\right)
,\left\vert u\right\vert ^{\widetilde{p}}\right\rangle \right] \left(
t\right) -\frac{1}{p}\left\langle q\left( x\right) ,\left\vert
u_{0}\right\vert ^{p}\right\rangle -\frac{1}{\widetilde{p}}\left\langle
a,\left\vert u_{0}\right\vert ^{\widetilde{p}}\right\rangle .
\end{equation*}

For estimate the $\overset{t}{\underset{0}{\dint }}\func{Re}\left\langle
f\left( u\right) ,\frac{\partial \overline{u}}{\partial s}+\overline{u}%
\right\rangle ds$ we use (1.2) (i.e. the condition (\textit{ii})) 
\begin{equation*}
\left\langle a\left( x\right) \left\vert u\right\vert ^{\widetilde{p}-2}u,%
\overline{u}\right\rangle \geq -k_{0}\left\Vert u\right\Vert
_{p_{2}}^{2}-k_{1}\left\langle q\left( x\right) ,\left\vert u\right\vert
^{p}\right\rangle
\end{equation*}%
then we obtain 
\begin{equation*}
\overset{t}{\underset{0}{\dint }}\func{Re}\left\langle f\left( u\right) ,%
\frac{\partial \overline{u}}{\partial s}+\overline{u}\right\rangle ds\geq 
\frac{1}{2}\left\Vert \nabla u\left( t\right) \right\Vert _{2}^{2}+\overset{t%
}{\underset{0}{\dint }}\left\Vert \nabla u\left( s\right) \right\Vert
_{2}^{2}ds+
\end{equation*}%
\begin{equation*}
\delta _{1}\overset{t}{\underset{0}{\dint }}\left\langle q\left( x\right)
\left\vert u\right\vert ^{p-2}u,\overline{u}\right\rangle \left( s\right)
ds-k_{0}\overset{t}{\underset{0}{\dint }}\left\Vert u\right\Vert
_{p_{2}}^{2}\left( s\right) ds+\delta _{2}\left\langle q\left( x\right)
\left\vert u\right\vert ^{p-2}u,\overline{u}\right\rangle \left( t\right) -
\end{equation*}%
\begin{equation*}
-k_{0}\left\Vert u\left( t\right) \right\Vert _{p_{2}}^{2}-\frac{1}{2}%
\left\Vert \nabla u_{0}\right\Vert _{2}^{2}-\frac{1}{p}\left\langle q\left(
x\right) ,\left\vert u_{0}\right\vert ^{p}\right\rangle -\frac{1}{\widetilde{%
p}}\left\langle a,\left\vert u_{0}\right\vert ^{\widetilde{p}}\right\rangle
\geq
\end{equation*}%
\begin{equation*}
\eta \left\Vert \nabla u\left( t\right) \right\Vert _{2}^{2}+\eta \overset{t}%
{\underset{0}{\dint }}\left\Vert \nabla u\left( s\right) \right\Vert
_{2}^{2}ds+\delta _{1}\overset{t}{\underset{0}{\dint }}\left\langle q\left(
x\right) \left\vert u\right\vert ^{p-2}u,\overline{u}\right\rangle \left(
s\right) ds+
\end{equation*}%
\begin{equation}
\delta _{2}\left\langle q\left( x\right) \left\vert u\right\vert ^{p-2}u,%
\overline{u}\right\rangle \left( t\right) -C\left( \left\Vert \nabla
u_{0}\right\Vert ,\left\Vert q\right\Vert _{W^{-1,2}},\left\Vert
u_{0}\right\Vert _{2^{\ast }},\left\Vert a\right\Vert _{m},p,\widetilde{p}%
\right)  \tag{4.6}
\end{equation}%
for a.e. $t>0$, where $\delta _{1}=1-k_{1}\geq 0$, $\delta _{2}=p^{-1}-%
\widetilde{p}^{-1}k_{1}\geq 0$, $\eta =1-C\left( 2,p_{2}\right) ^{2}k_{0}>0$%
. The expressions (4.5) and (4.6) shows that the condition (c) of the main
theorem takes place for the operator $f:X\longrightarrow Y$ generated by
posed problem.
\end{proof}

2. Now we prove an inequation used in the proof of the fulfilment of the
condition (d').

\begin{proposition}
Let all conditions of Theorem 1 are fulfilled, then the following inequality 
\begin{equation*}
\left\Vert f\left( u\right) -f\left( v\right) \right\Vert _{Y}\geq
\left\Vert u-v\right\Vert _{2}\left( t\right) +\left\Vert \nabla \left(
u-v\right) \right\Vert _{2}-
\end{equation*}%
\begin{equation*}
M\max \left\{ \left\Vert u\right\Vert _{\widetilde{p}}^{\widetilde{p}%
-2};\left\Vert v\right\Vert _{\widetilde{p}}^{\widetilde{p}-2}\right\}
\left\Vert u-v\right\Vert _{\widetilde{p}},
\end{equation*}%
holds for any $u,v\in X\cap \left\{ u\left\vert \ u\left( 0,x\right)
=u_{0}\left( x\right) \right. \right\} $.
\end{proposition}

\begin{proof}
Let us $u,v\in X\cap \left\{ u\left\vert \ u\left( 0,x\right) =u_{0}\left(
x\right) \right. \right\} $ and consider $\left\Vert f\left( u\right)
-f\left( v\right) \right\Vert _{Y}$ that we can estimate as 
\begin{equation*}
\left\Vert i\frac{\partial \left( u-v\right) }{\partial t}-\Delta \left(
u-v\right) +q\left( \left\vert u\right\vert ^{p-2}u-\left\vert v\right\vert
^{p-2}v\right) +a\left( \left\vert u\right\vert ^{\widetilde{p}%
-2}u-\left\vert v\right\vert ^{\widetilde{p}-2}v\right) \right\Vert _{Y}\geq
\end{equation*}%
\begin{equation*}
\left\Vert i\frac{\partial \left( u-v\right) }{\partial t}-\Delta \left(
u-v\right) +q\left( \left\vert u\right\vert ^{p-2}u-\left\vert v\right\vert
^{p-2}v\right) \right\Vert _{Y}-
\end{equation*}%
\begin{equation}
\left\Vert a\left( \left\vert u\right\vert ^{\widetilde{p}-2}u-\left\vert
v\right\vert ^{\widetilde{p}-2}v\right) \right\Vert _{Y}.  \tag{4.7}
\end{equation}%
In order that to esimate of the first adding of right side of the inequality
(4.7) we act in the following way. In beginning we set 
\begin{equation*}
\left\langle i\frac{\partial \left( u-v\right) }{\partial t},\overline{%
\left( u-v\right) }\right\rangle -\left\langle \Delta \left( u-v\right) ,%
\overline{\left( u-v\right) }\right\rangle +
\end{equation*}%
\begin{equation*}
\left\langle q\left( \left\vert u\right\vert ^{p-2}u-\left\vert v\right\vert
^{p-2}v\right) ,\overline{\left( u-v\right) }\right\rangle =\frac{i}{2}\frac{%
d}{dt}\left\langle \left( u-v\right) ,\overline{\left( u-v\right) }%
\right\rangle +
\end{equation*}%
\begin{equation}
\left\Vert \nabla \left( u-v\right) \right\Vert _{2}^{2}+\left\langle
q\left( \left\vert u\right\vert ^{p-2}u-\left\vert v\right\vert
^{p-2}v\right) ,\overline{\left( u-v\right) }\right\rangle  \tag{4.8}
\end{equation}%
and study it.

Here for the last adding takes place the inequality 
\begin{equation*}
\left\langle q\left( \left\vert u\right\vert ^{p-2}u-\left\vert v\right\vert
^{p-2}v\right) ,\overline{\left( u-v\right) }\right\rangle =\left\langle
q,\left\vert u\right\vert ^{p}\right\rangle +\left\langle q,\left\vert
v\right\vert ^{p}\right\rangle -
\end{equation*}%
\begin{equation*}
\left\langle q\left\vert u\right\vert ^{p-2}u,\overline{v}\right\rangle
-\left\langle q\left\vert v\right\vert ^{p-2}v,\overline{u}\right\rangle .
\end{equation*}%
As the expression $\left\vert \left\langle q\left\vert u\right\vert ^{p-2}u,%
\overline{v}\right\rangle +\left\langle q\left\vert v\right\vert ^{p-2}v,%
\overline{u}\right\rangle \right\vert $ has the following estimation 
\begin{equation*}
\left\vert \left\langle q\left\vert u\right\vert ^{p-2}u,\overline{v}%
\right\rangle +\left\langle q\left\vert v\right\vert ^{p-2}v,\overline{u}%
\right\rangle \right\vert \leq \left\langle q\left\vert u\right\vert
^{p-1},\left\vert v\right\vert \right\rangle +\left\langle q\left\vert
v\right\vert ^{p-1},\left\vert u\right\vert \right\rangle
\end{equation*}%
therefore we can consider the right side of (4.8) without of the last adding.

Consequently we get\ the following estimation for the first adding of the
right side of\ the inequality (4.7) 
\begin{equation*}
\left\Vert i\frac{\partial \left( u-v\right) }{\partial t}-\Delta \left(
u-v\right) +q\left( \left\vert u\right\vert ^{p-2}u-\left\vert v\right\vert
^{p-2}v\right) \right\Vert _{Y}\geq
\end{equation*}%
\begin{equation*}
K\left[ \left\Vert u-v\right\Vert _{2}\left( t\right) +\left\Vert \nabla
\left( u-v\right) \right\Vert _{2}\right] ,\ K>0
\end{equation*}%
with taking into account the equation (4.8), the last reasons and the
equation 
\begin{equation*}
\overset{t}{\underset{0}{\dint }}\left\langle \frac{\partial \left(
u-v\right) }{\partial s},\overline{\left( u-v\right) }\right\rangle ds=\frac{%
1}{2}\left\Vert u-v\right\Vert _{2}^{2}\left( t\right) ,
\end{equation*}%
whereas $u\left( 0,x\right) =v\left( 0,x\right) =u_{0}$ by choosingly, that
we need make by virtue of the condition (d') of the main theorem.

Now consider the second adding of right side of\ the inequality (4.7), for
which we have 
\begin{equation*}
\left\vert \left\langle a\left( \left\vert u\right\vert ^{\widetilde{p}%
-2}u-\left\vert v\right\vert ^{\widetilde{p}-2}v\right) ,\overline{\left(
u-v\right) }\right\rangle \right\vert =\underset{\Omega }{\dint }a\left(
\left\vert u\right\vert ^{\widetilde{p}-2}u-\left\vert v\right\vert ^{%
\widetilde{p}-2}v\right) \overline{\left( u-v\right) }dx\leq
\end{equation*}%
\begin{equation*}
\underset{\Omega }{\dint }a\ \varphi \left( u,v\right) \ \left\vert
u-v\right\vert ^{2}dx,\quad 0\leq \varphi \left( u,v\right) \leq M\ \left(
\max \left\{ \left\vert u\right\vert ,\left\vert v\right\vert \right\}
\right) ^{\widetilde{p}-2}
\end{equation*}%
where $M>0$ be some number and $\varphi \left( u,v\right) $ be a continuous
function.

Taking into account the last inequalities in (4.7) we obtain 
\begin{equation*}
\left\Vert f\left( u\right) -f\left( v\right) \right\Vert _{Y}\geq
\left\Vert u-v\right\Vert _{2}\left( t\right) +\left\Vert \nabla \left(
u-v\right) \right\Vert _{2}-
\end{equation*}%
\begin{equation}
M\max \left\{ \left\Vert u\right\Vert _{\widetilde{p}}^{\widetilde{p}%
-2};\left\Vert v\right\Vert _{\widetilde{p}}^{\widetilde{p}-2}\right\}
\left\Vert u-v\right\Vert _{\widetilde{p}}.  \tag{4.9}
\end{equation}
\end{proof}

3. Now we will conduct a priori estimations for a solutions of the problem.

\begin{proposition}
Let all conditions of Theorem 1 are fulfilled, then all solutions belong to
bounded subset of the space 
\begin{equation*}
X\equiv W^{1,2}\left( R_{+};L^{2}\left( \Omega \right) \right) \cap
L^{m}\left( R_{+};W_{0}^{1,2}\left( \Omega \right) \right) \cap
\end{equation*}%
\begin{equation*}
\left\{ v\left\vert \ \left\vert v\right\vert ^{p}\right. \in L^{\beta
}\left( R_{+};W_{0}^{1,\beta }\left( \Omega \right) \right) \right\} \cap
\left\{ u\left\vert \ u\left( 0,x\right) =u_{0}\left( x\right) \right.
\right\} ,
\end{equation*}%
i.e. there is constants $K\equiv K\left( \left\Vert h\right\Vert
_{2,Q_{T}},\left\Vert u_{0}\right\Vert _{W^{1,2}},\left\Vert q\right\Vert
_{W^{-1,2}},\left\Vert a\right\Vert ,p,\widetilde{p}\right) $ such that $%
\left\Vert u\right\Vert _{X}\leq K.$
\end{proposition}

\begin{proof}
In the beginning we note that from examination of the expression $%
\left\langle f\left( u\right) ,\overline{u}\right\rangle $ in the proof of
Proposition 2 we get the following inequations 
\begin{equation*}
\frac{1}{2}\left\Vert u\left( t\right) \right\Vert _{2}^{2}-\frac{1}{2}%
\left\Vert u_{0}\right\Vert _{2}^{2}\leq \overset{t}{\underset{0}{\dint }}%
\left\vert \left\langle h,\overline{u}\right\rangle \right\vert ds\leq 
\overset{t}{\underset{0}{\dint }}\left\Vert h\left( s\right) \right\Vert
_{2}\left\Vert u\left( s\right) \right\Vert _{2}ds
\end{equation*}%
and 
\begin{equation*}
\left\Vert \nabla u\right\Vert _{2}^{2}+\left\langle q\left( x\right)
\left\vert u\right\vert ^{p-2}u,\overline{u}\right\rangle +\left\langle
a\left( x\right) \left\vert u\right\vert ^{\widetilde{p}-2}u,\overline{u}%
\right\rangle \leq \left\vert \left\langle h,\overline{u}\right\rangle
\right\vert \Longrightarrow
\end{equation*}%
\begin{equation*}
\left\Vert \nabla u\right\Vert _{2}^{2}+\left( 1-k_{1}\right) \left\langle
q\left( x\right) ,\left\vert u\right\vert ^{p}\right\rangle -k_{0}\left\Vert
u\right\Vert _{p_{2}}^{2}\leq \left\vert \left\langle h,\overline{u}%
\right\rangle \right\vert \Longrightarrow
\end{equation*}%
\begin{equation*}
\left( 1-k_{0}C^{2}\left( p_{2}\right) \right) \left\Vert \nabla u\left(
t\right) \right\Vert _{2}^{2}\leq \left\Vert h\left( t\right) \right\Vert
_{2}\left\Vert u\left( t\right) \right\Vert _{2},
\end{equation*}%
as $k_{1}\leq 1$ and $k_{0}C^{2}\left( p_{2}\right) <1$ by virtue of the
condition (\textit{ii}). Then we get that the following estimations are true 
\begin{equation*}
\left( 1-k_{0}C^{2}\left( p_{2}\right) \right) \left\Vert \nabla u\left(
t\right) \right\Vert _{2}^{2}-\varepsilon \left\Vert u\left( t\right)
\right\Vert _{2}^{2}\leq c\left( \varepsilon \right) \left\Vert h\left(
t\right) \right\Vert _{2}^{2}
\end{equation*}%
or 
\begin{equation*}
\left\Vert \nabla u\left( t\right) \right\Vert _{2}\leq \widehat{c}\left(
\varepsilon \right) \left\Vert h\left( t\right) \right\Vert _{2},\ a.e.t\geq
0
\end{equation*}%
as far as $\left\Vert u\left( t\right) \right\Vert _{2}\leq C_{1}\left( mes\
\Omega \right) \left\Vert \nabla u\left( t\right) \right\Vert _{2}$ for $%
\forall u\left( t\right) \in W_{0}^{1,2}\left( \Omega \right) $ by the
embedding theorems, where $\widehat{c}\left( \varepsilon \right) \equiv 
\widehat{c}\left( \varepsilon ,mes\ \Omega ,k_{0}C\right) $, $C_{1}\left(
mes\ \Omega \right) >0$ are constants;

moreover 
\begin{equation}
\left\Vert u\right\Vert _{L^{2}\left( R_{+};W_{0}^{1,2}\left( \Omega \right)
\right) }\leq \widehat{c}_{1}\left( \varepsilon \right) \left\Vert
h\right\Vert _{L^{2}\left( R_{+};W_{0}^{1,2}\left( \Omega \right) \right) }+%
\widehat{c}_{2}\left\Vert u_{0}\right\Vert _{2}.  \tag{*}
\end{equation}

Thus we obtain that $u\left( t,x\right) $ belong to the bounded subset of $%
L^{2}\left( R_{+};W_{0}^{1,2}\left( \Omega \right) \right) $ for given $h\in
L^{2}\left( Q\right) $.

Using the equations (4.3) and (4.4), and also the estimates (4.5) and (4.6)
we get 
\begin{equation*}
\left\vert \overset{t}{\underset{0}{\dint }}\func{Re}\left\langle h,\frac{%
\partial \overline{u}}{\partial s}+\overline{u}\right\rangle ds\right\vert
\geq \overset{t}{\underset{0}{\dint }}\func{Re}\left\langle f\left( u\right)
,\frac{\partial \overline{u}}{\partial s}+\overline{u}\right\rangle ds\geq
\end{equation*}%
\begin{equation*}
\eta _{2}\left\Vert \nabla u\left( t\right) \right\Vert _{2}^{2}+\eta _{1}%
\overset{t}{\underset{0}{\dint }}\left\Vert \nabla u\left( s\right)
\right\Vert _{2}^{2}ds+\delta _{1}\overset{t}{\underset{0}{\dint }}%
\left\langle q\left( x\right) \left\vert u\right\vert ^{p-2}u,\overline{u}%
\right\rangle \left( s\right) ds+
\end{equation*}%
\begin{equation}
\delta _{2}\left\langle q\left( x\right) \left\vert u\right\vert ^{p-2}u,%
\overline{u}\right\rangle \left( t\right) -C\left( \left\Vert
u_{0}\right\Vert _{W^{1,2}},\left\Vert q\right\Vert _{W^{-1,2}},\left\Vert
a\right\Vert _{m},p,\widetilde{p}\right)  \tag{4.10}
\end{equation}%
and 
\begin{equation*}
\left\vert \overset{t}{\underset{0}{\dint }}\func{Im}\left\langle h,\frac{%
\partial \overline{u}}{\partial s}+\overline{u}\right\rangle ds\right\vert
\geq \overset{t}{\underset{0}{\dint }}\func{Im}\left\langle f\left( u\right)
,\frac{\partial \overline{u}}{\partial s}+\overline{u}\right\rangle \equiv
\end{equation*}%
\begin{equation}
\frac{1}{2}\left\Vert u\left( t\right) \right\Vert _{2}^{2}-\frac{1}{2}%
\left\Vert u_{0}\right\Vert _{2}^{2}+\overset{t}{\underset{0}{\dint }}%
\left\Vert \frac{\partial u}{\partial s}\right\Vert _{L_{2}\left( \Omega
\right) }^{2}ds  \tag{4.11}
\end{equation}

Then from (4.10) we get the estimation 
\begin{equation*}
\eta _{2}\left\Vert \nabla u\left( t\right) \right\Vert _{2}^{2}+\eta _{1}%
\overset{t}{\underset{0}{\dint }}\left\Vert \nabla u\left( s\right)
\right\Vert _{2}^{2}ds+\delta _{1}\overset{t}{\underset{0}{\dint }}%
\left\langle q\left( x\right) \left\vert u\right\vert ^{p-2}u,\overline{u}%
\right\rangle \left( s\right) ds+
\end{equation*}%
\begin{equation*}
\delta _{2}\left\langle q\left( x\right) \left\vert u\right\vert ^{p-2}u,%
\overline{u}\right\rangle \left( t\right) -C\left( \left\Vert
u_{0}\right\Vert _{W^{1,2}},\left\Vert q\right\Vert _{W^{-1,2}},\left\Vert
a\right\Vert _{m^{\prime }},p,\widetilde{p}\right) \leq
\end{equation*}%
\begin{equation*}
\overset{t}{\underset{0}{\dint }}\left\Vert h\right\Vert _{2}\left(
\left\Vert \frac{\partial u}{\partial s}\right\Vert _{2}+\left\Vert
u\right\Vert _{2}\right) ds,\text{ \ \ }a.e.\ t>0,
\end{equation*}%
and from (4.11) we obtain 
\begin{equation*}
\frac{1}{2}\left\Vert u\left( t\right) \right\Vert _{2}^{2}-\frac{1}{2}%
\left\Vert u_{0}\right\Vert _{2}^{2}+\overset{t}{\underset{0}{\dint }}%
\left\Vert \frac{\partial u}{\partial s}\right\Vert _{L_{2}\left( \Omega
\right) }^{2}ds\leq
\end{equation*}%
\begin{equation*}
\overset{t}{\underset{0}{\dint }}\left\Vert h\right\Vert _{2}\left(
\left\Vert \frac{\partial u}{\partial s}\right\Vert _{2}+\left\Vert
u\right\Vert _{2}\right) ds,\text{ \ \ }a.e.\ t>0,
\end{equation*}%
then with combine of last two inequations we get 
\begin{equation*}
\eta \left\Vert \nabla u\left( t\right) \right\Vert _{2}^{2}+\eta \overset{t}%
{\underset{0}{\dint }}\left\Vert \nabla u\left( s\right) \right\Vert
_{2}^{2}ds+\delta _{1}\overset{t}{\underset{0}{\dint }}\left\langle q\left(
x\right) \left\vert u\right\vert ^{p-2}u,\overline{u}\right\rangle \left(
s\right) ds+
\end{equation*}%
\begin{equation*}
\delta _{2}\left\langle q\left( x\right) \left\vert u\right\vert ^{p-2}u,%
\overline{u}\right\rangle \left( t\right) -C\left( \left\Vert
u_{0}\right\Vert _{W^{1,2}},\left\Vert q\right\Vert _{W^{-1,2}},\left\Vert
a\right\Vert ,p,\widetilde{p}\right) +
\end{equation*}%
\begin{equation*}
\frac{1}{2}\left\Vert u\left( t\right) \right\Vert _{2}^{2}-\frac{1}{2}%
\left\Vert u_{0}\right\Vert _{2}^{2}+\overset{t}{\underset{0}{\dint }}%
\left\Vert \frac{\partial u}{\partial s}\right\Vert _{2}^{2}ds\leq
\varepsilon _{1}\overset{t}{\underset{0}{\dint }}\left\Vert \frac{\partial u%
}{\partial s}\right\Vert _{2}^{2}ds+
\end{equation*}%
\begin{equation*}
\varepsilon _{2}\overset{t}{\underset{0}{\dint }}\left\Vert u\right\Vert
_{2}^{2}ds+C\left( \varepsilon _{1},\varepsilon _{2}\right) \overset{t}{%
\underset{0}{\dint }}\left\Vert h\right\Vert _{2}^{2}ds,\text{ \ \ }a.e.\
t>0.
\end{equation*}

or 
\begin{equation*}
\eta \left\Vert \nabla u\left( t\right) \right\Vert _{2}^{2}+\widetilde{\eta 
}\overset{t}{\underset{0}{\dint }}\left\Vert \nabla u\left( s\right)
\right\Vert _{2}^{2}ds+\delta _{1}\overset{t}{\underset{0}{\dint }}%
\left\langle q\left( x\right) \left\vert u\right\vert ^{p-2}u,\overline{u}%
\right\rangle \left( s\right) ds+
\end{equation*}%
\begin{equation*}
\delta _{2}\left\langle q\left( x\right) \left\vert u\right\vert ^{p-2}u,%
\overline{u}\right\rangle \left( t\right) +\frac{1}{2}\left\Vert u\left(
t\right) \right\Vert _{2}^{2}+\left( 1-\varepsilon _{1}\right) \overset{t}{%
\underset{0}{\dint }}\left\Vert \frac{\partial u}{\partial s}\right\Vert
_{2}^{2}ds\leq
\end{equation*}%
\begin{equation*}
C\left( \varepsilon _{1},\varepsilon _{2}\right) \overset{\infty }{\underset{%
0}{\dint }}\left\Vert h\right\Vert _{2}^{2}ds+C_{1}\left( \left\Vert
u_{0}\right\Vert _{W^{1,2}},\left\Vert q\right\Vert _{W^{-1,2}},\left\Vert
a\right\Vert ,p,\widetilde{p}\right) .
\end{equation*}

Moreover from here follows 
\begin{equation*}
\frac{1}{2}\left\Vert \nabla u\left( t\right) \right\Vert _{2}^{2}-\frac{1}{2%
}\left\Vert \nabla u\left( 0\right) \right\Vert _{2}^{2}+\overset{t}{%
\underset{0}{\dint }}\left\Vert \nabla u\left( s\right) \right\Vert
_{2}^{2}ds+
\end{equation*}%
\begin{equation*}
\overset{t}{\underset{0}{\dint }}\left\langle q\left( x\right) \left\vert
u\right\vert ^{p-2}u,\overline{u}\right\rangle \left( s\right) ds+\overset{t}%
{\underset{0}{\dint }}\left\langle a\left( x\right) \left\vert u\right\vert
^{\widetilde{p}-2}u,\overline{u}\right\rangle \left( s\right) ds+
\end{equation*}%
\begin{equation*}
\frac{1}{p}\left\langle q\left( x\right) ,\left\vert u\right\vert
^{p}\right\rangle \left( t\right) -\frac{1}{p}\left\langle q\left( x\right)
,\left\vert u\left( 0\right) \right\vert ^{p}\right\rangle +\frac{1}{%
\widetilde{p}}\left\langle a,\left\vert u\right\vert ^{\widetilde{p}%
}\right\rangle \left( t\right) -
\end{equation*}%
\begin{equation*}
\frac{1}{\widetilde{p}}\left\langle a,\left\vert u\left( 0\right)
\right\vert ^{\widetilde{p}}\right\rangle \leq \left( \varepsilon ^{-1}+%
\frac{1}{2}\right) \overset{t}{\underset{0}{\dint }}\left\Vert h\left(
s\right) \right\Vert _{2}^{2}ds+
\end{equation*}%
\begin{equation*}
\varepsilon \overset{t}{\underset{0}{\dint }}\left\Vert u\left( s\right)
\right\Vert _{2}^{2}ds+\frac{1}{2}\overset{t}{\underset{0}{\dint }}%
\left\Vert \frac{\partial u}{\partial s}\right\Vert _{2}^{2}ds+
\end{equation*}%
\begin{equation*}
C\left( \left\Vert \nabla u_{0}\right\Vert ,\left\Vert q\right\Vert
_{W^{-1,2}},\left\Vert u_{0}\right\Vert _{2^{\ast }},\left\Vert a\right\Vert
_{m},p,\widetilde{p}\right) ,
\end{equation*}%
by virtue of the condition (1.2).

Consequently we get the following inequation 
\begin{equation*}
\eta \left\Vert \nabla u\left( t\right) \right\Vert _{2}^{2}+\eta _{1}\left(
\varepsilon \right) \overset{t}{\underset{0}{\dint }}\left\Vert \nabla
u\left( s\right) \right\Vert _{2}^{2}ds\leq \varepsilon ^{-1}\overset{t}{%
\underset{0}{\dint }}\left\Vert h\left( s\right) \right\Vert _{2}^{2}ds+
\end{equation*}%
\begin{equation*}
C\left( \left\Vert \nabla u_{0}\right\Vert _{2},\left\Vert q\right\Vert
_{W^{-1,2}},\left\Vert u_{0}\right\Vert _{2^{\ast }},\left\Vert h\right\Vert
_{2},\left\Vert a\right\Vert _{m},p,\widetilde{p}\right)
\end{equation*}%
in the other words we obtain 
\begin{equation}
\overset{t}{\underset{0}{\dint }}\left\Vert \nabla u\left( s\right)
\right\Vert _{2}^{2}ds\leq \widetilde{D}\left( \varepsilon ^{-1},...\right)
\left( \widetilde{\eta }_{1}\left( \varepsilon \right) \right) ^{-1}\left(
1-e^{-\widetilde{\eta }_{1}\left( \varepsilon \right) t}\right)  \tag{4.12}
\end{equation}%
where $\widetilde{D}\left( \varepsilon ^{-1},...\right) =$ $\widetilde{D}%
\left( \varepsilon ^{-1},\left\Vert u_{0}\right\Vert _{W^{1,2}},\left\Vert
q\right\Vert _{W^{-1,2}},\left\Vert h\right\Vert _{2},\left\Vert
a\right\Vert _{m},p,\widetilde{p}\right) $

Consequently we obtain, that any solution of the considered problem under
the posed conditions satisfies the following inclusion 
\begin{equation*}
u\in W^{1,2}\left( R_{+};L^{2}\left( \Omega \right) \right) \cap L^{\infty
}\left( R_{+};W_{0}^{1,2}\left( \Omega \right) \right) \cap
\end{equation*}%
\begin{equation}
\left\{ v\left\vert \ \left\vert v\right\vert ^{p}\right. \in L^{\beta
}\left( R_{+};W_{0}^{1,\beta }\left( \Omega \right) \right) \right\} \cap
\left\{ u\left\vert \ u\left( 0,x\right) =u_{0}\left( x\right) \right.
\right\} \equiv X,  \tag{4.13}
\end{equation}%
in addition the preimage of each bounded neighborhood of zero from $%
L^{2}\left( Q\right) \times W_{0}^{1,2}\left( \Omega \right) $ under
operator $f$ is the bounded neighborhood of zero of the space determined by
(4.13), where $m\geq 2^{\ast }$ and $\beta >1$ is denoted by Lemma 1. We
note that here is used the inclution given in next remark.
\end{proof}

\begin{remark}
Let $Z$ is a Banach space, then 
\begin{equation*}
L^{2}\left( R;Z\right) \cap L^{\infty }\left( R;Z\right) \subset L^{m}\left(
R;Z\right) ,\ 2\leq m<\infty
\end{equation*}%
holds.
\end{remark}

From here we get that the condition (c) is fulfilled for the operator $f$
generated by the posed problem and the operator $g\left( v\right) \equiv 
\frac{\partial v}{\partial t}+v$ for any $v\in W^{1,2}\left(
R_{+};W_{0}^{1,2}\left( \Omega \right) \right) \cap L^{m}\left(
R_{+};W_{0}^{1,2}\left( \Omega \right) \right) $ that is dense in the
requisit space defined in (4.13).

4. Now we can show that the operator $f$ satisfies the condition (d').

\begin{proposition}
Let all conditions of Theorem 1 are fulfilled, then the operator $f$
satisfies the condition (d').
\end{proposition}

\begin{proof}
More exactly, we will prove that the image $f(X)$ is the closed subset of $Y$%
. As the conditions (a), (b) and (c) of the main theorem are fulfilled for
the operator $f:X\rightarrow Y$ then we get, that $f(X)$ contains a dense
subset of the space $Y$ by virtue of the first statement of this theorem.

So, let us the sequence $\left\{ h_{k}\right\} _{k=1}^{\infty }\subset f(X)$
is the fundamental sequence in $Y$ that converge to an element $h_{0}\in Y$,
since $f(X)$ contains a dense subset of the space $Y$ therefore for any $%
h_{0}\in Y$ there exists a sequence of such type. As the sequence $\left\{
h_{k}\right\} _{k=1}^{\infty }$ is a bounded subset of $Y$, and consequently 
$f^{-1}\left( \left\{ h_{k}\right\} _{k=1}^{\infty }\right) $ belong to the
bounded subset $M_{0}$ of $X$ by virtue of the condition (c), which is
proved in the step 3. It is known that $X$ is the reflexive space therefore
we can choose a subsequence $\left\{ u_{k_{j}}\right\} _{j=1}^{\infty
}\subset M_{0}$ of $f^{-1}\left( \left\{ h_{k}\right\} _{k=1}^{\infty
}\right) $ such that $u_{k_{j}}\in f^{-1}\left( h_{k_{j}}\right) $, $%
k_{j}\nearrow \infty $, and $\left\{ u_{k_{j}}\right\} _{j=1}^{\infty }$
weakly converge in $X$, i.e. $u_{k_{j}}\rightharpoonup u_{0}\in X$. Moreover
it is known that $X\Subset L^{m}\left( R_{+};L^{\ell }\left( \Omega \right)
\right) $ is compact, where $1<\ell <\frac{2n}{n-2}$ if $n\geq 3$ (see, for
example, [30] and \ its references). Then the sequence $\left\{
u_{k_{j}}\right\} _{j=1}^{\infty }$ have a subsequence, that strongly
converge in the space $L^{m}\left( R_{+};L^{\ell }\left( \Omega \right)
\right) $, which for simplicity we denote also by $\left\{ u_{k_{j}}\right\}
_{j=1}^{\infty }$, i.e. we assume that $\left\{ u_{k_{j}}\right\}
_{j=1}^{\infty }$ is the subsequence of such type.

Thence use previous reasons and the (4.9) from step 2 we get $%
u_{k_{j}}\Longrightarrow u_{0}$ in $L^{2}\left( R_{+};W_{0}^{1,2}\left(
\Omega \right) \right) $ and in $L^{m}\left( R_{+};L^{2}\left( \Omega
\right) \right) $. Furthermore if take into account that in this problem
first two adding are linear continuous operator then we obtain that $%
h_{k_{j}}=f\left( u_{k_{j}}\right) \longrightarrow f\left( u_{0}\right)
\equiv h_{0}$, which show that $f(X)$ is the closed subset of $Y$.
\end{proof}

\textbf{Thus we can complete of the proof of Theorem 1}. From Propositions
2-5 we get, that the operator $f:X\longrightarrow Y$ generated by the posed
problem satisfies all conditions of the main theorem (Theorem 2, and also
Theorem 3). Then using Theorem 2 we obtain, that the operator $f$ satisfies
the statement of Theorem 2, therefore the statement of Theorem 1 is correct.
Consequently, the existence theorem for the problem (0.1)-(0.2) (i.e.
Theorem 1) is proved.

\section{Behaviour of Solutions of Problem (0.1) - (0.2)}

We will study the behaviour of the solution of problem (0.1) - (0.2) in the
sense of the space $W^{1,2}\left( Q\right) \cap $ $L^{m}\left(
R_{+};W_{0}^{1,2}\left( \Omega \right) \right) $. Consider the following
functional on the space $W^{1,2}\left( Q\right) $ $\cap $ $L^{m}\left(
R_{+};W_{0}^{1,2}\left( \Omega \right) \right) $ \ 
\begin{equation*}
I\left( v\left( t\right) \right) =\left\Vert \nabla v\left( t\right)
\right\Vert _{L^{2}}^{2}\equiv \left\Vert \nabla v\left( t\right)
\right\Vert _{2}^{2}\equiv \underset{\Omega }{\dint }\left\vert \nabla
v\left( t,x\right) \right\vert ^{2}dx.
\end{equation*}

So we have 
\begin{equation*}
2^{-1}\frac{d}{dt}I\left( u\left( t\right) \right) =\left\langle \frac{%
\partial }{\partial t}\nabla u,\overline{\nabla u}\right\rangle
=-\left\langle \frac{\partial }{\partial t}u,\overline{\Delta u}%
\right\rangle =
\end{equation*}%
here if to take account $u\left( t,x\right) $ is the solution of the
considered problem then we get 
\begin{equation*}
-i\left\langle \frac{\partial u}{\partial t},\frac{\partial \overline{u}}{%
\partial t}\right\rangle -\left\langle \frac{\partial u}{\partial t}%
,q\left\vert u\right\vert ^{p-2}\overline{u}\right\rangle -\left\langle 
\frac{\partial u}{\partial t},a\left\vert u\right\vert ^{\widetilde{p}-2}%
\overline{u}\right\rangle +\left\langle \frac{\partial u}{\partial t},%
\overline{h}\right\rangle =
\end{equation*}%
\begin{equation*}
-i\left\Vert \frac{\partial u}{\partial t}\right\Vert _{2}^{2}-p^{-1}\frac{d%
}{dt}\left\langle q,\left\vert u\right\vert ^{p}\right\rangle -\widetilde{p}%
^{-1}\frac{d}{dt}\left\langle a,\left\vert u\right\vert ^{\widetilde{p}%
}\right\rangle +\left\langle \frac{\partial u}{\partial t},\overline{h}%
\right\rangle
\end{equation*}%
whence we have 
\begin{equation*}
2^{-1}\frac{d}{dt}I\left( u\left( t\right) \right) =-\left\langle \frac{%
\partial u}{\partial t},q\left\vert u\right\vert ^{p-2}\overline{u}%
\right\rangle -\left\langle \frac{\partial u}{\partial t},a\left\vert
u\right\vert ^{\widetilde{p}-2}\overline{u}\right\rangle +\func{Re}%
\left\langle \frac{\partial u}{\partial t},\overline{h}\right\rangle
\end{equation*}%
and 
\begin{equation}
-\left\Vert \frac{\partial u}{\partial t}\right\Vert _{2}^{2}=-\func{Im}%
\left\langle \frac{\partial u}{\partial t},\overline{h}\right\rangle
\Longrightarrow \left\Vert \frac{\partial u}{\partial t}\right\Vert _{2}\leq
\left\Vert h\right\Vert _{2}  \tag{5.1}
\end{equation}%
for a.e. $t>0,$ as $I\left( u\left( t\right) \right) $ is a real function.

Thus we obtain 
\begin{equation*}
2^{-1}I\left( u\left( t\right) \right) -2^{-1}I\left( u_{0}\right) \leq
-p^{-1}\overset{t}{\underset{0}{\dint }}\frac{d}{ds}\left\langle
q,\left\vert u\right\vert ^{p}\right\rangle ds-\widetilde{p}^{-1}\overset{t}{%
\underset{0}{\dint }}\frac{d}{ds}\left\langle a,\left\vert u\right\vert ^{%
\widetilde{p}}\right\rangle ds+
\end{equation*}%
\begin{equation*}
\overset{t}{\underset{0}{\dint }}\left\vert \left\langle \frac{\partial u}{%
\partial s},\overline{h}\right\rangle \right\vert ds\leq -p^{-1}\left\langle
q,\left\vert u\right\vert ^{p}\right\rangle \left( t\right)
+p^{-1}\left\langle q,\left\vert u_{0}\right\vert ^{p}\right\rangle -
\end{equation*}%
\begin{equation*}
\widetilde{p}^{-1}\left\langle a,\left\vert u\right\vert ^{\widetilde{p}%
}\right\rangle \left( t\right) +\widetilde{p}^{-1}\left\langle a,\left\vert
u_{0}\right\vert ^{\widetilde{p}}\right\rangle +\overset{t}{\underset{0}{%
\dint }}\left\Vert \frac{\partial u}{\partial s}\right\Vert _{2}\left\Vert
h\right\Vert _{2}\left( s\right) ds\leq
\end{equation*}%
\begin{equation*}
k_{0}\left\Vert u\left( t\right) \right\Vert _{p_{2}}^{2}-\delta
\left\langle q,\left\vert u\right\vert ^{p}\right\rangle \left( t\right) +%
\overset{t}{\underset{0}{\dint }}\left\Vert h\right\Vert _{2}^{2}\left(
s\right) ds
\end{equation*}%
whence we get using the condition (1.2') (i.e. the inequality $%
2^{-1}>C\left( 2,p_{2}\right) ^{2}\cdot k_{0}$) and (5.1) 
\begin{equation*}
I\left( u\left( t\right) \right) \equiv \left\Vert \nabla u\left( t\right)
\right\Vert _{2}^{2}\leq 2k_{0}\left\Vert u\left( t\right) \right\Vert
_{p_{2}}^{2}+2\overset{t}{\underset{0}{\dint }}\left\Vert h\right\Vert
_{2}^{2}\left( s\right) ds+I\left( u_{0}\right) \Longrightarrow
\end{equation*}%
\begin{equation}
\left( 1-2C\left( 2,p_{2}\right) ^{2}k_{0}\right) \left\Vert \nabla u\left(
t\right) \right\Vert _{2}^{2}\leq 2\overset{t}{\underset{0}{\dint }}%
\left\Vert h\right\Vert _{2}^{2}\left( s\right) ds+I\left( u_{0}\right) 
\tag{5.2}
\end{equation}%
for a.e. $t\geq 0$.

Moreover in the previous section we obtained the following equations 
\begin{equation*}
I_{0}\left( u\left( t\right) \right) \equiv \left\Vert u\right\Vert
_{2}^{2}\Longrightarrow \frac{1}{2}\frac{d}{dt}I_{0}\left( u\left( t\right)
\right) =\left\langle u_{t},\overline{u}\right\rangle =
\end{equation*}%
\begin{equation*}
\left\langle u_{t},\overline{u}\right\rangle =i\left\langle -\Delta
u+q\left\vert u\right\vert ^{p-2}u+a\left\vert u\right\vert ^{\widetilde{p}%
-2}u-h,\overline{u}\right\rangle =
\end{equation*}%
\begin{equation*}
=i\left\Vert \nabla u\left( t\right) \right\Vert _{2}^{2}+i\left\langle
q\left( x\right) ,\left\vert u\right\vert ^{p}\right\rangle +i\left\langle
a\left( x\right) ,\left\vert u\right\vert ^{\widetilde{p}}\right\rangle
-i\left\langle h,\overline{u}\right\rangle \Longrightarrow
\end{equation*}%
\begin{equation*}
\left\Vert \nabla u\left( t\right) \right\Vert _{2}^{2}+\left\langle q\left(
x\right) ,\left\vert u\right\vert ^{p}\right\rangle +\left\langle a\left(
x\right) ,\left\vert u\right\vert ^{\widetilde{p}}\right\rangle -\func{Re}%
\left\langle h,\overline{u}\right\rangle =0
\end{equation*}%
\begin{equation*}
\&\ \frac{1}{2}\frac{d}{dt}I_{0}\left( u\left( t\right) \right) =-\func{Im}%
\left\langle h,\overline{u}\right\rangle
\end{equation*}%
whence follows 
\begin{equation}
\left\Vert \nabla u\left( t\right) \right\Vert _{2}^{2}\leq 4\left\Vert
h\left( t\right) \right\Vert _{2}^{2},\quad a.e.t\geq 0.  \tag{5.3}
\end{equation}

Consequently is proved the following result.

\begin{theorem}
Let the solution $u\left( t,x\right) $ of the problem (0.1)-(0.2) is a
sufficiently smooth function, i.e. $u\in W^{1,2}\left(
R_{+};W_{0}^{1,2}\left( \Omega \right) \right) \cap $ $L^{m}\left(
R_{+};W_{0}^{1,2}\left( \Omega \right) \right) $ and all conditions of
Theorem 1 are fulfilled. Then if the coefficient $k_{0}$ of the condition
(1.2) satisfy the inequality $2^{-1}>C\left( 2,p_{2}\right) ^{2}\cdot k_{0}$
then the following inequalities are true 
\begin{equation}
\left\Vert \nabla u\left( t\right) \right\Vert _{2},\left\Vert \frac{%
\partial u}{\partial t}\right\Vert _{2}\leq b\left\Vert h\left( t\right)
\right\Vert _{2}\Longrightarrow \left\Vert u\left( t\right) \right\Vert
_{2}\leq b_{1}\left\Vert h\left( t\right) \right\Vert _{2}  \tag{5.4}
\end{equation}%
for a.e. $t\geq 0$.
\end{theorem}

\begin{notation}
It should be noted that the next inequality follows from (4.12) in the
conditions of Theorem 5.1 
\begin{equation*}
\left\Vert \nabla u\left( t\right) \right\Vert _{2}^{2}\leq \widetilde{D}%
\left( \varepsilon ^{-1},...\right) \ \exp \left\{ -\widetilde{\eta }%
_{1}\left( \varepsilon \right) t\right\} +\left\Vert \nabla u_{0}\right\Vert
_{2}^{2}.)
\end{equation*}
\end{notation}

\end{document}